\documentclass[twocolumn,10pt]{IEEEtran}
\usepackage{braket}

\input{def.tex}
\usepackage{dsfont}
\usepackage{minibox}
\DeclareSymbolFont{matha}{OML}{txmi}{m}{it}
\DeclareMathSymbol{\varv}{\mathord}{matha}{118}
\usepackage{eurosym}
\usepackage{hhline,physics}

\usepackage{multicol}
\usepackage{algorithm,algorithmicx}
\usepackage{graphicx}
\usepackage{textcomp}
\usepackage{caption,pifont}
\usepackage[multiple]{footmisc}
\makeatletter

\IEEEoverridecommandlockouts
\begin{document}
	\title{Intelligent Reflecting Surface-assisted MU-MISO Systems with Imperfect Hardware: Channel Estimation  and Beamforming Design}
	\author{Anastasios Papazafeiropoulos, Cunhua Pan, Pandelis Kourtessis, Symeon Chatzinotas, John M. Senior \thanks{A. Papazafeiropoulos is with the Communications and Intelligent Systems Research Group, University of Hertfordshire, Hatfield AL10 9AB, U. K., and with SnT at the University of Luxembourg, Luxembourg. C. Pan with the School of Electronic Engineering and Computer Science at Queen Mary University of London, London E1 4NS, U.K. P. Kourtessis and John M.Senior are with the Communications and Intelligent Systems Research Group, University of Hertfordshire, Hatfield AL10 9AB, U. K. S. Chatzinotas is with the SnT at the University of Luxembourg, Luxembourg. A. Papazafeiropoulos was supported  by the University of Hertfordshire's 5-year Vice Chancellor's Research Fellowship.
		S. Chatzinotas   was supported by the National Research Fund, Luxembourg, under the projects RISOTTI. E-mails: tapapazaf@gmail.com, , c.pan@qmul.ac.uk, \{p.kourtessis,j.m.senior\}@herts.ac.uk, symeon.chatzinotas@uni.lu.}}
	\maketitle\vspace{-1.7cm}
	\begin{abstract}
		Intelligent reflecting surface (IRS), consisting of low-cost passive elements, is a promising technology for improving the spectral and energy efficiency of the fifth-generation (5G) and beyond networks. It is also noteworthy that an IRS  can shape the reflected signal propagation. Most works in IRS-assisted systems have ignored the impact of the inevitable residual hardware impairments (HWIs) at both the transceiver hardware and the IRS while any relevant works have addressed only simple scenarios, e.g., with single-antenna network nodes and/or without taking the randomness of phase noise at the IRS into account. In this work, we aim at filling up this gap by considering a general IRS-assisted multi-user (MU) multiple-input single-output (MISO) system with imperfect channel state information (CSI) and correlated Rayleigh fading. In parallel, we present a  general computationally efficient methodology for IRS reflecting beamforming (RB) optimization. Specifically, we introduce an advantageous channel estimation (CE) method for such systems accounting for the HWIs. Moreover, we derive the uplink achievable spectral efficiency (SE) with maximal-ratio combining (MRC) receiver, displaying three significant advantages being: 1) its closed-form expression, 2) its dependence only on large-scale statistics, and 3) its low training overhead. Notably, by exploiting the first two benefits, we achieve to perform optimization with respect to the RB that can take place only per several coherence intervals, and thus, reduces significantly the computational cost compared to other methods based on instantaneous CSI which require frequent phase optimization. Among the  insightful observations, we highlight  that the unrealistic assumption of  uncorrelated Rayleigh fading does not allow optimization of the SE, which makes the application of an IRS ineffective. Also, in the case that the phase drifts, describing the distortion of the phases in the RBM, are uniformly distributed, the presence of an IRS provides no advantage.  The analytical results outperform previous works and are verified by Monte-Carlo (MC) simulations.
	\end{abstract}
	
	\section{Introduction}
	In the last decade, a variety of wireless technological advances, including millimeter-wave (mmWave) communication and massive multiple-input multiple-output (mMIMO) systems, have been proposed to achieve a $ 1000 $-fold capacity increase and ubiquitous wireless connectivity among a large number of devices \cite{Boccardi2014}. Unfortunately, these technologies face practical limitations in terms of excessive energy consumption and hardware cost as well as no guaranteed quality of service (QoS) in harsh propagation environments. For instance, mmWave communications exhibit high penetration/path loss while requiring expensive and energy-consuming transceivers. Similarly, mMIMO systems manifest low performance in poor scattering conditions while the large  number of active elements might render the energy usage prohibitive. Moreover, new challenging use cases will emerge with possibly similar shortcomings. As a result,  future networks require  radical  paradigm shifts towards their energy sustainability, e.g.,  control to some extent over the propagation environment.
	
	A disruptive technology, covering this gap, has emerged under the label intelligent reflecting surfaces (IRSs) or  reconfigurable intelligent surfaces (RISs). An IRS consists of a meta-surface including a large number of reconfigurable passive elements that are able to function independently by inducing certain phase shifts on the impinging waves \cite{Basar2019}. The smart adjustment of the phase shifts is managed by an attached controller and allows the coherent addition of the reflected signals to boost the desired signal at the receiver.
	The IRS design and applications have attracted a lot of significant research interest \cite{Basar2019,Wu2019a,Pan2020,Bjoernson2019b,Kammoun2020, Elbir2020,Guo2020,Chen2019}. For example, in \cite{Wu2019a}, the downlink of an IRS-assisted multi-user (MU) multiple-input single-output (MISO) communication system was studied by jointly optimizing the precoding and reflecting beamforming matrices (RBMs), in order to minimize the transmit power at the base station (BS) with signal-to-interference-plus-noise ratio (SINR) constraints. In addition, the sum-rate was maximized in \cite{Pan2020} subject to a transmit power constraint. 
	Furthermore, the outperformance of the IRS with respect to the decode-and-forward (DF) relay was presented in \cite{Bjoernson2019b}. Also, the authors in \cite{Kammoun2020} maximized the minimum user rate in the large number of antennas regime. Despite the fundamental design issues, applications regarding IRSs have started to emerge such as the maximization of the minimum secrecy rate for physical layer security \cite{Chen2019}. 
	
	  In general, there are two approaches for phases optimization as the literature reveals.  The first  method is based on statistical channel state information (CSI) \cite{Jia2020,Zhao2020,Zhi2020,Kammoun2020,Papazafeiropoulos2021,VanChien2021,Papazafeiropoulos2021a} and the second method is based on instantaneous CSI \cite{Wu2019a,Pan2020}. According to the second method, the phases are optimized at every coherence
		interval since the corresponding expressions depend on small-scale channel fading. On the contrary, the first method includes expressions that depend on the large-scale statistics, which change at every several coherence intervals. Thus, the 	significance of the first method is noteworthy since it reduces considerably the signal overhead 	which can be prohibitive in the case of a large number of reflecting elements at the IRS. Moreover,
		it results in lower computational complexity.  Especially, based on these observations, in high-speed  scenarios with fast time-varying channels,  it is more practical to design and adjust the IRS phase shifts  according to the statistical
		CSI   while the tuning of the IRS 	parameters, based on instantaneous CSI, would be more challenging since they would have to
		be updated more frequently. Furthermore, although the IRS does not consume ideally transmit 	power, its smart controller is power-consuming and its continuous overloading with operations in 	the case of instantaneous CSI would not be energy efficient.
	
	Although most existing works with IRS-aided systems have relied on the knowledge of perfect CSI, this is a highly unrealistic assumption. In practice, systems have imperfect CSI. Especially, their passive elements make them energy efficient, but, contrary to conventional systems, this interesting feature does not enable them to accomplish the channel estimation (CE) task  by transmission/reception of pilot symbols. Hence, it is of paramount importance to take into account the CE before arriving at realistic conclusions. Among the fundamental works \cite{Zheng2019,Mishra2019,He2019,Elbir2020,Nadeem2020}, the authors in \cite{Mishra2019} proposed an ON/OFF channel	estimation scheme that obtains one-by-one least squares (LS) estimates of all IRS-assisted channels for a single-user MISO system. Moving to MUs systems, finding more applications in contemporary systems, the authors in \cite{He2019} exploited the sparsity of the channel and formulated a sparse channel matrix recovery problem for CE. In \cite{Nadeem2020}, the authors extended the model in \cite{Mishra2019} by assuming all IRS elements to be active during training while a number of sub-phases equal at least to the number of IRS elements is considered. This method provides better CE as the number of sub-phases increases, but the achievable rate worsens since the data transmission phase takes a smaller fraction of the coherence time due to excessive training overhead. Another drawback is that this method provides the estimates of the channels of the individual IRS elements while the covariance of the channel vector from all IRS elements to a specific user equipment (UE) is unknown. 

	On the other hand, prior literature of IRS-assisted systems has mostly assumed perfect hardware while practical applications are affected by unavoidable transceiver hardware impairments (T-HWIs) such as the in-phase/quadrature-phase (I/Q)-imbalance~\cite{Qi2010}, the quantization noise in the analog-to-digital converters (ADCs), and the oscillator phase noise (PN)~\cite{Papazafeiropoulos2016,Papazafeiropoulos2017a}. Even if mitigation/compensation algorithms exist, T-HWIs cannot be completely removed \cite{Schenk2008,Bjoernson2017}. Basically, T-HWIs are divided into two main categories being the additive and multiplicative T-HWIs. 
	In this work, we focus on the impact of the additive T-HWIs, while the study of multiplicative T-HWIs will be the topic for future work. In this direction, an examination of existing works with HWIs in the area of IRS-assisted systems shows that relevant studies are in their infancy \cite{Li2020,Xing2020,Qian2020,Liu2020,Shen2020,Zhou2020a}. In \cite{Li2020}, only single-antennas nodes were considered, and the phase errors were assumed known (deterministic). In addition, the phase noise, induced by an IRS and, henceforth called IRS-HWIs, has been studied in \cite{Xing2020,Qian2020} in the case of perfect CSI, but no expectation was taken over the phase noise.  Note that this phase noise, coming from the finite precision configuration of the phase drifts, is irrelevant with the phase noise coming from imperfect signal generation in  local oscillators in standard antenna systems.  Moreover, in \cite{Liu2020}, despite its randomness, again, no averaging of the phase noise was applied. Furthermore, only a single UE communication has been considered and only upper bounds on the channel capacities have been studied, which are not also obtained in closed forms. The authors in \cite{Shen2020} provided the beamforming optimization by accounting for T-HWIs in a single-user setting, and in \cite{Zhou2020a}, the secrecy rate was derived. Notably, only a few works in IRS-aided systems have assumed correlated Rayleigh fading despite that this is normally the case in multi-antenna next-generation systems. Apart from that, most works perform RBM optimization with a high computational cost in every coherence interval.

	\subsection{Contribution}
	The previous observations motivate the topic of this work, which is the design/study of a general IRS-assisted MU-MISO system with imperfect CSI and HWIs at both the IRS and the transceiver while performing robust optimization. Notably, the introduction of HWIs increases the complexity/difficulty and demands substantial manipulations during the analysis of IRS-assisted systems. The main contributions are summarized as follows:
	\begin{itemize}
		\item Contrary to \cite{Xing2020}, we have assumed multiple antennas at the BS and multiple UEs as well as imperfect CSI. Also, compared to \cite{Liu2020}, we have considered correlated Rayleigh fading, multiple UEs, and closed-form lower bounds, which are more practical than any upper bounds. Both references have not addressed properly the impact of phase noise while, in \cite{Li2020}, only deterministic phase noise was assumed. In \cite{Shen2020} and \cite{Zhou2020a}, only the impact of T-HWIs was studied while only a single destination and perfect CSI were assumed.    Notably, as far as the authors are aware, our work is the only one accounting for the randomness of the phase noise.
		\item 		 Many previous works have assumed that the optimization of the RBM should take place at every coherence interval since the corresponding expressions depend on small-scale channel fading, while our proposed results, being dependent only on large-scale statistics are suggested to be optimized at every several coherence intervals. Thus, their significance is noteworthy since they reduce considerably the signal overhead which can be prohibitive in the case of  a large number of reflecting elements at the IRS.
		\item We perform CE by means of linear minimum mean square error (LMMSE) while HWIs are taken into account. In parallel, we have assumed correlated Rayleigh fading. Our method provides analytical tractable expressions with low overhead compared to previous works.\footnote{Works such as \cite{He2019, Elbir2020} do not provide analytical expressions. Also, 
			previous CE methods with correlated fading do not allow the derivation of an optimizable achievable spectral efficiency (SE) being dependent on the RBM. In \cite{Nadeem2020}, only the estimated individual channels between each IRS element and each UE are obtained while the inter-element correlation is unknown.}
		\item We derive the uplink achievable spectral efficiency (SE) (lower bound) of an IRS-assisted MU-MISO system with MRC, imperfect CSI, and HWIs in a closed-form dependent only on   large-scale statistics (covariances).
		\item We optimize the achievable sum SE with respect to the RBM. As mentioned, contrary to other works that depend on small-scale statistics (e.g., see \cite{Guo2020} where the stochastic successive convex approximation technique has been performed), our optimization can be performed quite efficiently by the project gradient ascent at every several coherence intervals since both the sum SE and the proposed algorithm require only the large-scale statistics.   Notably, contrary to existing  works such as \cite{Jia2020,Zhao2020}, based on statistical CSI, we achieve to provide the SE and the phases optimization in closed-form. 
		\item We shed light on the degradation of the uplink sum SE of an IRS-aided MU-MISO system   due to the presence of imperfect CSI, HWIs, and correlated fading. For example, we thoroughly examine how the probability density function (PDF) of the phase noise at the IRS and the severity of the T-HWIs affect the performance of  IRS-aided systems.
	\end{itemize}
	
	\subsection{Paper Outline} 
	The remainder of this paper is organized as follows. Section~\ref{System} presents the system model of an IRS-assisted MU-MISO system with correlated Rayleigh fading and HWIs. Section~\ref{ChannelEstimation} provides the CE. Section~\ref{PerformanceAnalysis} presents the uplink sum SE and the optimization concerning the IRS RBM.
	The numerical results are placed in Section~\ref{Numerical}, and Section~\ref{Conclusion} concludes the paper.
	
	\subsection{Notation}Vectors and matrices are denoted by boldface lower and upper case symbols, respectively. The notations $(\cdot)^\T$, $(\cdot)^\H$, and $\tr\!\left( {\cdot} \right)$ represent the transpose, Hermitian transpose, and trace operators, respectively. The expectation operator is denoted by $\EE\left[\cdot\right]$ (or $\EE_{x}\left[\cdot\right]$ to denote expectation with respect to $ x $) while $ \diag\left(\ba \right) $ represents an $ n\times n $ diagonal matrix with diagonal elements being the elements of vector $ \ba $. In the case of a matrix $ \bA $, $\diag\left(\bA\right) $ denotes a diagonal matrix with elements corresponding to the diagonal elements of $ \bA $. Also, $ \arg\left(\cdot\right) $ and $ \circ $ denote the argument function and the Hadamard product, respectively. Finally, $\bb \sim \cC\cN{(\b0,\mathbf{\Sigma})}$ represents a circularly symmetric complex Gaussian vector with {zero mean} and covariance matrix $\mathbf{\Sigma}$.
	
	\section{System Model}\label{System}
	We consider an IRS-aided MU-MISO system  as depicted in Fig. \ref{Fig10}. In particular, a  BS, equipped with $ M $ antennas, serves $K $ single-antenna UEs by means of one IRS consisting of $ N $ passive reflecting elements introducing phase shifts onto the incoming signal waves. The phase-shifts are adjusted by a controller exchanging information with the BS through a backhaul link. Reasonably, the IRS is deployed in the line-of-sight (LoS) of the BS  by assuming that  both the BS and IRS are deployed at high altitude and their locations are fixed. Moreover, the proposed model assumes direct links between the BS and the UEs, but these could be neglected in certain scenarios.  For example, in mmWave transmission, suggested by 5G and beyond systems, high penetration losses and resultant signal blockages do not allow the presence of an LoS component \cite{Rappaport2019}.
	\begin{figure}[!h]
		\begin{center}
			\includegraphics[width=0.9\linewidth]{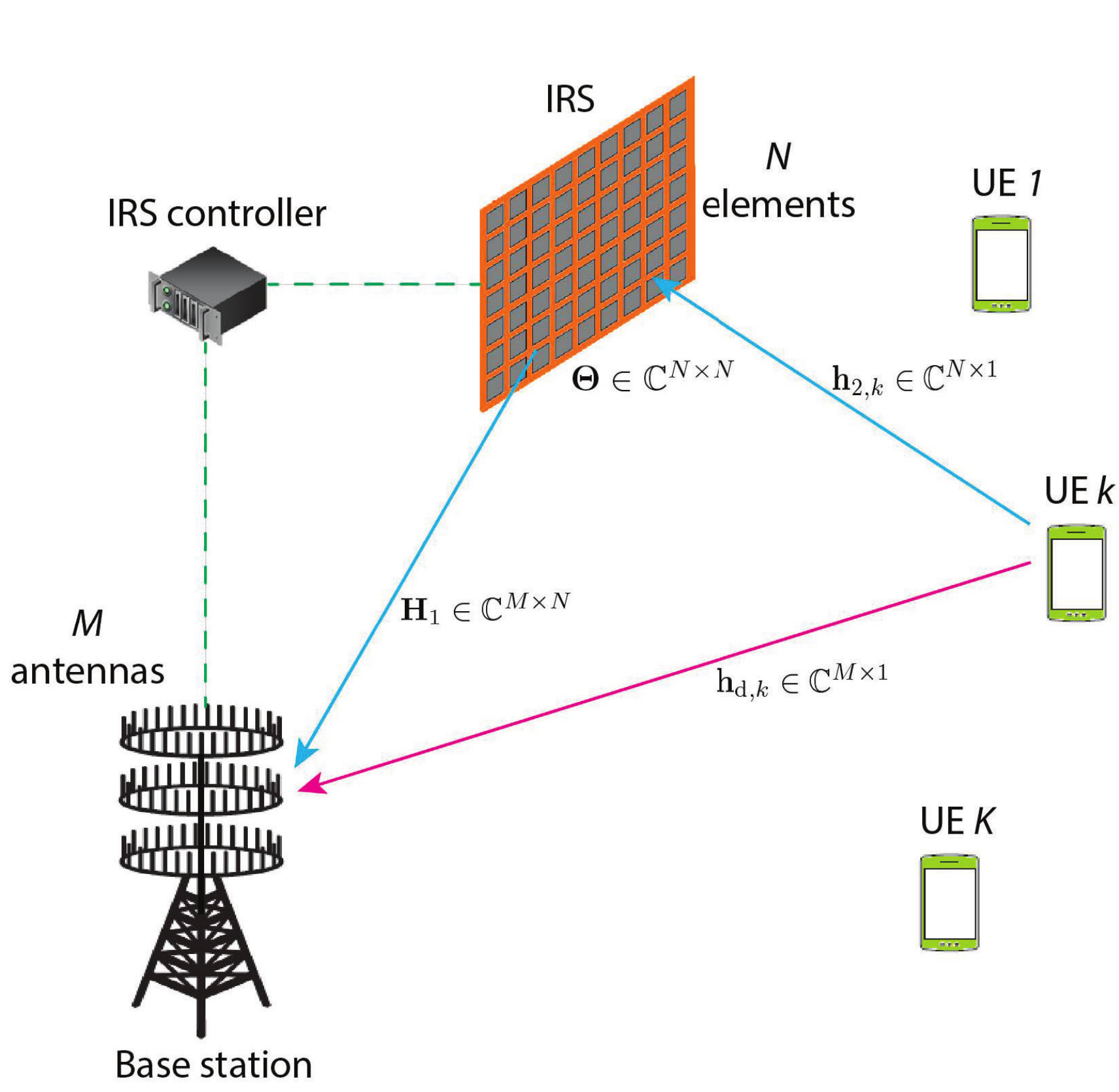}
			\caption{\footnotesize{ An IRS-assisted uplink MU-MISO communication system with $ M $ BS antennas, $ N $ IRS elements, and $ K $ UEs.  }}
			\label{Fig10}
		\end{center}
	\end{figure}
	
	\subsection{Channel Model}\label{ChannelModel} 
	We assume a time-varying narrowband channel, divided into coherence blocks, where each block has a duration $\tau_{\mathrm{c}}$ channel uses.\footnote{ The extension to the wideband case could follow the lines of \cite{Zheng2019,Yang2020} and could be the topic of future work.} Note that $\tau_{\mathrm{c}}= B_{\mathrm{c}}T_{\mathrm{c}}$ with $ B_{\mathrm{c}} $ and $ T_{\mathrm{c}} $ being the coherence bandwidth and the coherence time in $\mathrm{Hz}$ and $\mathrm{s}$, respectively. Especially, we employ the standard time-division-duplex (TDD) protocol, where each block accounts for $\tau$ channel uses for the uplink training phase and $\tau_{\mathrm{u}}=\tau_{\mathrm{c}}-\tau$ channel uses for the uplink data transmission phase.
	
	Let $ \bh_{\mathrm{d},k} \in \mathbb{C}^{M \times 1} $, $ \bH_{1}=[\bh_{1,1}\ldots,\bh_{1,N} ] \in \mathbb{C}^{M \times N}$, and $ \bh_{2,k} \in \mathbb{C}^{N \times 1}$ be the direct channel between the BS and UE $ k $, the LoS channel between the BS and the IRS with $ \bh_{1,i} $ for $ i=1,\ldots,N $ being its column vectors, and the channel between the IRS and UE $ k $. The subscripts $ 1$ and $ 2$ correspond to the BS-IRS and IRS-UE $ k $ links, respectively. Although the majority of existing works, e.g., \cite{Wu2019a,Pan2020}, assumed independent Rayleigh model, we account for spatial correlation,
	which appears in practice and affects the performance \cite{Kammoun2020}. Hence, $ \bh_{\mathrm{d},k} $ and $ \bh_{2,k} $ are described in terms of correlated Rayleigh fading distributions as
	\begin{align}
		\bh_{2,k}&=\sqrt{\beta_{2,k}}\bR_{\mathrm{IRS},k}^{1/2}\bz_{k},\\
		\bh_{\mathrm{d},k}&=\sqrt{\beta_{\mathrm{d},k}}\bR_{\mathrm{BS},k}^{1/2}\bz_{\mathrm{d},k},
	\end{align}
	where $ \bR_{\mathrm{IRS},k} \in \mathbb{C}^{N \times N} $ and $ \bR_{\mathrm{BS},k} \in \mathbb{C}^{M \times M} $ describe the deterministic Hermitian-symmetric positive semi-definite correlation
	matrices at the IRS and the BS respectively with $ \tr\left(\bR_{\mathrm{IRS},k} \right)=N $ and $ \tr\left(\bR_{\mathrm{BS},k} \right)=M $. Notably, the correlation matrices $ \bR_{\mathrm{IRS},k} $ and $ \bR_{\mathrm{BS},k}~\forall k$ are assumed to be known by the network. They can be obtained through the existing estimation methods (see e.g., \cite{Neumann2018}). Given that correlation models for IRS using meta-surfaces are not known, we adopt the correlation model in \cite{Kammoun2015} for conventional antenna arrays.\footnote{While writing this work,  the authors in \cite{Bjoernson2020} presented a more suitable correlation for IRSs. Its thorough study is the topic of ongoing research.} Also, $ \beta_{2,k} $ and $ \beta_{\mathrm{d},k} $ denote the path-loss  of the IRS-UE $ k $ and BS-UE $ k $ links, respectively. Furthermore, $ \bz_{k}\sim \mathcal{CN}\left(\b0,\Id_{N}\right) $ and $ \bz_{\mathrm{d},k} \sim \mathcal{CN}\left(\b0,\Id_{N}\right) $ describe the corresponding fast-fading vectors.
	
	By taking into account that the IRS is designed to be installed in a location providing an LoS channel with the BS, the channel matrix $ \bH_{1} $ will likely have rank one, which results in performance gains only when $ K=1 $ \cite{Kammoun2020}. However, an MU scenario requires $ \mathrm{rank}\left(\bH_{1}\right)\ge K $. The higher rank could be achieved by placing the IRS close to the BS or by 
	deterministic scattering between the BS and the IRS. The high rank LoS channel $ \bH_{1} $ can be obtained as
	\begin{align}
		[\bH_{1}]_{m,n}=\sqrt{\beta_{1}} \exp\big(j \frac{2 \pi }{\lambda}\left(m-1\right)d_{\mathrm{BS}}\sin \theta_{1,n}\sin \phi_{1,n}\nn\\
		+\left(n-1\right)d_{\mathrm{IRS}}\sin \theta_{2,m}\sin \phi_{2,m}\big)\!,
	\end{align}
	where $ \lambda $ is the carrier wavelength, $ \beta_{1} $ is the path-loss between the BS and IRS while $ d_{\mathrm{BS}} $ and $ d_{\mathrm{IRS}} $ are the inter-antenna separation at the BS and inter-element separation at the IRS, respectively \cite{Nadeem2020}. Also, $ \theta_{1,n} $, $ \phi_{1,n} $ describe the elevation and azimuth LoS angles of departure (AoD) at the BS with respect to IRS element $ n $, and $ \theta_{2,n} $, $ \phi_{2,n} $ describe the elevation and azimuth LoS angles of arrival (AoA) at the IRS.  In practice, $ d_{\mathrm{BS}} $ and $ d_{\mathrm{IRS}} $ are known by construction while  the angles, depending only on the locations, can be calculated when the locations are given according to \cite{Hu2020}.  It is worthwhile to mention that the estimation of the correlation matrices could also be obtained similar to $ \bH_1 $ since the dependence of their expressions on the distances and the angles is similar. 
	
	\subsection{Ideal Uplink Signal Model}\label{SignalModel}
	The ideal received complex baseband signal vector by the BS is written as
	\begin{align}
		\by= & \sum_{k =1}^{K} \left(\bh_{\mathrm{d},k}+\bH_{1}\bTheta \bh_{2,k}\right)x_{k}+ 
		\bw,\label{eq:Ypt1}
	\end{align}
	where $ \bw \sim \cC\cN\left(\b0,\sigma^2\Id_{M}\right) $ is the additive white Gaussian noise (AWGN) vector at the BS while $\bTheta=\mathrm{diag}\left( \al_{1}e^{j \theta_{1}}, \ldots, \al_{N}e^{j \theta_{N}} \right)\in\mathbb{C}^{N\times N}$ is the RBM being diagonal and representing the response of the $ N $ elements with $ \theta_{n} \in [0,2\pi]$ and $ \al_{n}\in [0,1] $ denoting the phase  and amplitude coefficient for element $ n $, respectively. As commonly assumed due to recent advances towards lossless metasurfaces \cite{Epstein2016,Badloe2017}, we set $ \al_{n}=1~ \forall n$, i.e., we assume maximum signal reflection. For the sake of exposition, given the RBM, we denote the overall channel vector $ \bh_{k}=\bh_{\mathrm{d},k}+ \bH_{1}\bTheta \bh_{2,k} $, distributed as $ \bh_{k}\sim \cC\cN\left( 0, \bR_{k} \right) $, where $ \bR_{k}= \beta_{\mathrm{d},k}\bR_{\mathrm{BS},k}+ \beta_{2,k}\bH_{1} \bTheta {\bR}_{\mathrm{IRS},k}\bTheta^{\H}\bH_{1}^{\H}$.
	
	\subsection{Hardware Impairments}\label{HardwareImpairments} 
	In this work, we consider two distinct types of HWIs  in an IRS-assisted system: 1) the aggregate additive HWIs at the transceiver, and 2) the HWIs emerged from the passive elements of the IRS. Henceforth, we  denote them T-HWIs and IRS-HWIs, respectively.
	\subsubsection{T-HWIs}
	The majority of papers in the IRS literature have relied on the unrealistic assumption of ideal transceiver hardware.  Especially, next-generation antenna deployments with a large number of antennas such as a massive MIMO systems assisted by an IRS should be implemented with cheap hardware, in order to be cost-efficient as the number of antennas increases. However, cheaper hardware results in lower quality with more severe HWIs that are more power consuming and degrade further the system performance. Instead, we take into account  the additive distortions at both the transmitter and the receiver being Gaussian distributed with average powers proportional to the average transmit and received signals, respectively \cite{Schenk2008}.  We would like to mention that although Gaussian modeling for the HWIs can be assumed  rudimentary, it is used widely because its tractability allows extracting primary conclusions, e.g., see the recent works \cite{Zhang2020,Papazafeiropoulos2021b}. The Gaussianity results by means of the aggregate contribution of many impairments. Especially, let $ p_{k}=\EE \{|x_{k}|^{2}\}$ be the transmit power from UE $ k $ having transmit signal $ x_{k} $ and $ \bh_{k} \in \mathbb{C}^{M\times 1}$ be the channel vector of this UE. The additive transceiver distortions during the uplink are described in terms of conditional distributions with respect to the channel realizations as
	\begin{align}
		\delta_{\mathrm{t},k}&\sim \cC\cN\left( 0, \Lambda_{k} \right),\label{eta_tU} \\
		\deltav_{\mathrm{r}}&\sim \cC\cN \left( \b0,\bm \Upsilon \right)\label{eta_rU},
	\end{align}
	where $ \Lambda_{k}= \kappa_{\mathrm{UE}}p_{k}$ and 
	$\bm \Upsilon =\kappa_{\mathrm{BS}}\sum_{i=1}^{K}p_{i}$   $ \mathrm{diag}\left( |h_{i,1}|^{2},\ldots,|h_{i,M}|^{2} \right) $ with $ \bh_{i}=\left[h_{i,1},\ldots,h_{i,M}\right]^{\T} $. The variance $ \bm \Upsilon $ can also be written as
	$\bm \Upsilon =\kappa_{\mathrm{BS}}\sum_{i=1}^{K}p_{i} \Id_{M}\circ \bh_{i} \bh_{i}^{\H}$. The proportionality parameters $\kappa_{\mathrm{UE}}$ and $\kappa_{\mathrm{BS}}$ express the severity of the residual impairments at the transmitter and receiver side, and are met in applications in terms of the error vector magnitude (EVM)~\cite{Holma2011}.  For example, the EVM at the BS is defined as 
		\begin{align}
			\mathrm{EVM}_{\mathrm{BS}}=\sqrt{\frac{\EE[\|\deltav_{\mathrm{r}}\|^{2}_{2}]}{\EE[\|\by\|_{2}^{2}]}}=\sqrt{\frac{\mathrm{tr}\left(\bm \Upsilon \right)}{\EE[\|\by\|_{2}^{2}]}}=		\sqrt{\kappa_{\mathrm{BS}}}, 
		\end{align}
		where the expectations take place for a specific channel realization. Here, for the sake of simplicity, we assume that the parameter $ \kappa_{\mathrm{UE}} $ is identical for all UEs.
	\subsubsection{IRS-HWIs} Taking into account that the reflection phases of the IRS passive elements cannot be configured with infinite precision, they can be modeled in terms of phase errors \cite{Badiu2019}.\footnote{We focus on the main IRS impairment, being the imperfection of phases configuration by assuming unity reflection amplitude, i.e., full signal reflection as in prior works e.g., see \cite{Basar2019,Wu2019a,Pan2020}. However, recently, it was suggested that the reflection amplitude can be  phase-dependent due to hardware limitations \cite{Abeywickrama2020}, which requires a separate analysis and is left for future work.} In particular, IRS-HWIs are mathematically described by means of a random diagonal phase error matrix consisting of $ N $ random phase errors, i.e., $ \widetilde{\bTheta} =\diag\left( e^{j \tilde{\theta}_{1}}, \ldots, e^{j \tilde{\theta}_{N}} \right)\in\mathbb{C}^{N\times N}$, where $ \tilde{\theta}_{i}, i=1,\ldots,N $ are the random phase errors being i.i.d. randomly distributed in $ [-\pi, \pi) $ according to a certain circular distribution. Also, we assume that the PDF of $\tilde{ \theta}_{i} $ is symmetric and its mean direction  is zero, i.e., $ \arg\left(\EE[\mathrm{e}^{j \tilde{\theta}_{i}}]\right)=0 $ \cite{Badiu2019}. The most widely used PDFs, being able to describe the phase noise are the uniform and the Von Mises distributions \cite{Badiu2019}, where:
	\begin{itemize}
		\item the uniform distribution expresses completely lack of knowledge (random reflection) and its characteristic function (CF) denoted by $ m $ is $0 $,
		\item the Von Mises distribution with a zero-mean and concentration parameter $ \kappa_{\tilde{\theta}} $, 	capturing the accuracy of the estimation, has a CF $ m= \frac{\mathrm{I}_{1}\!\left(\kappa_{\tilde{\theta}}\right)}{\mathrm{I}_{0}\!\left(\kappa_{\tilde{\theta}}\right)}$ with $ \mathrm{I}_{p}\!\left(\kappa_{\tilde{\theta}}\right)$ being the modified Bessel function of the first kind and order
		$ p $.
	\end{itemize} 
	\subsection{Realistic Uplink Signal Model with HWIs}\label{SignalModel1}
	Overall, the realistic received signal vector by the BS after having incorporated both the T-HWIs and IRS-HWIs in \eqref{eq:Ypt1} is given by
	\begin{align}
		\by= & \sum_{k =1}^{K} \left(\bh_{\mathrm{d},k}+\bH_{1}\bTheta\widetilde{\bTheta} \bh_{2,k}\right)\left(x_{k} +\delta_{\mathrm{t},k}\right)+\deltav_{\mathrm{r}}+
		\bw\label{eq:Ypt}.
	\end{align}
	
	Now, given the RBM, the overall channel vector is $ \bh_{k}=\bh_{\mathrm{d},k}+ \bH_{1}\bTheta\widetilde{\bTheta} \bh_{2,k} $, distributed as $ \bh_{k}\sim \cC\cN\left( 0, \bR_{k} \right) $, where $ \bR_{k}= \beta_{\mathrm{d},k}\bR_{\mathrm{BS},k}+ \beta_{2,k}\bH_{1} \bTheta \widetilde{\bR}_{\mathrm{IRS},k}\bTheta^{\H}\bH_{1}^{\H}$ with $ \widetilde{\bR}_{\mathrm{IRS},k} $ given by
	\begin{align}
		&\widetilde{\bR}_{\mathrm{IRS},k}=\EE[ \widetilde{\bTheta}\bR_{\mathrm{IRS},k}\widetilde{\bTheta}^{\H}]\\
		&=\!\!\begin{bmatrix}\!
			r_{11} &\!\!\!\! r_{12} \EE_{\tilde{\theta}}[e^{j \tilde{\theta}_{1}-j \tilde{\theta}_{2}}]&\!\!\!\! \dots&\!\!\!\!r_{1N} \EE_{\tilde{\theta}}[e^{j \tilde{\theta}_{1}-j \tilde{\theta}_{N}}]\\
			r_{21} \EE_{\tilde{\theta}}[e^{j \tilde{\theta}_{2}-j \tilde{\theta}_{1}} ]&\!\!\!\! r_{22} & \!\!\!\!\dots&\!\!\!\!r_{2N} \EE_{\tilde{\theta}}[e^{j \tilde{\theta}_{2}-j \tilde{\theta}_{N}}]\\
			\vdots&\!\!\!\!\vdots&\!\!\!\!\ddots&\!\!\!\!\vdots\\
			r_{N1} \EE_{\tilde{\theta}}[e^{j \tilde{\theta}_{N}-j \tilde{\theta}_{1}} ]&\!\!\!\! r_{N2} \EE_{\tilde{\theta}}[e^{j \tilde{\theta}_{N}-j \tilde{\theta}_{2}} ] &\!\!\!\! \dots&\!\!\!\!r_{NN}\!\end{bmatrix} \label{cor1}\\
		&=\!\begin{bmatrix}
			r_{11} & m^{2} r_{12} & \dots&m^{2} r_{1N}\\
			m^{2} r_{21}& r_{22} & \dots&m^{2} r_{2N} \\
			\vdots&\vdots&\ddots&\vdots\\
			m^{2} r_{N1} &m^{2} r_{N2} & \dots&r_{NN}\end{bmatrix} \label{cor2}\\
		&=m^{2}\bR_{\mathrm{IRS},k}+\left(1-m^{2}\right)\diag\left(\bR_{\mathrm{IRS},k}\right)\\
		&=m^{2}\bR_{\mathrm{IRS},k}+\left(1-m^{2}\right)\Id_{N}.\label{cor3}
	\end{align}
	In \eqref{cor1}, $ r_{ij} $ is the $ ij $th element of the correlation matrix $ \bR_{\mathrm{IRS},k} $. Also, in \eqref{cor2}, we have exploited that $ \tilde{\theta}_{i} $ are i.i.d. distributed with a symmetric PDF while we have substituted with $ m $ the corresponding CF. The next equation is written in a compact form in terms of $ \bR_{\mathrm{IRS},k} $. In \eqref{cor3}, we have accounted for that in correlated Rayleigh fading, we have $ r_{ii}=1, ~\forall i $. Notably, this is a very useful equation describing the dependence on the IRS-HWIs, i.e., the phase noises from the IRS elements.
	\begin{remark}\label{rem1}
		In the case of the uniform distribution, where the characteristic function is zero ($ m=0 $), we obtain $ \widetilde{\bR}_{\mathrm{IRS},k}=\Id_{N}$, i.e., there is no dependence on the phase noises. In such case, the overall covariance becomes $ \bR_{k} =\beta_{\mathrm{d},k}\bR_{\mathrm{BS},k}+\beta_{2,k}\bH_{1}\bH_{1}^{\H}$, which obviously has no dependence on the RBM and cannot be optimized. Hence, the IRS cannot be exploited. Moreover, no knowledge of $ \bR_{\mathrm{IRS},k} $ is required at the BS.  However, even in this case, the IRS still contributes with an additional signal to the receiver. Although it is expected to be weak, it is beneficial. Especially, when there is no direct signal.  Note that  these cases are very difficult to appear in practice. Specifically, regarding the independent Rayleigh assumption, it is uncommon to appear as mentioned in \cite{Bjoernson2020} while always there will be some knowledge and control in the reflection at the IRS, which means that the uniform distribution is not meaningful in practice (see Sec. II.C.2). On the contrary, any other circular PDF for the description of the phase errors allows studying the impact of these errors, and mostly, taking advantage of the IRS.
	\end{remark}

	\section{CE with HWIs}\label{ChannelEstimation}
	In practice, a BS does not have perfect CSI but estimates its channel by a TDD operation including an uplink training phase with pilot symbols \cite{Bjoernson2017}. Differently to conventional MISO with/without relay systems, the IRS implemented by means of passive elements, is not able to send pilots to the BS for CE or process the received pilot symbols from the UEs to obtain the corresponding estimated channels. Contrary to existing works providing separately the estimated direct and cascaded channels \cite{Nadeem2020}, we provide the estimate of the overall channel. For example, compared to \cite{Mishra2019} and \cite{Nadeem2020}, we perform the CE in a single phase instead of $ N+1 $ phases. In particular, the former is known as ON/OFF channel
	estimation and addresses a single UE setting while the latter assumes multiple UEs. Hence, the achievable SE in our case is higher since the pre-log factor in the SE is lower (lower training overhead). Also, we achieve to derive the covariance of the estimated cascaded channel vector while, in \cite{Nadeem2020}, it is assumed unknown.
	
	The CE protocol assumes that the total uplink training phase has a duration of $ \tau $ sec. Let the UEs transmit orthogonal pilot sequences. Especially, we denote by $\bx_{p,k}=[x_{p,k,1}, \ldots, x_{p,k,\tau}]^{\T}\in \mathbb{C}^{\tau\times 1} $ the pilot sequence of UE $ k $ with $ \bx_{p,k}^{\H}\bx_{p,l}=0~\forall k\ne l$ and $ \bx_{p,k}^{\H}\bx_{p,k}= \tau P$ joules, where $ P =|x_{p,k,i}|^{2} ,~\forall k,i$ is the common average transmit power per UE during the training phase.
	
	The received signal at the BS with T-HWIs and IRS-HWIs during the training period is given by
	\begin{align}
		\bY^{\tr}=\sum_{i=1}^{K}\bh_{i}\!\left(\bx_{\mathrm{p},i}^{\H} +\deltav_{\mathrm{t},i}^{\H} \right)+\bDelta_{\mathrm{r}}+
		\bW^{\tr},\label{train1}
	\end{align}
	where $ \deltav_{\mathrm{t},i} \in \mathbb{C}^{M \times 1} \sim \mathcal{CN}\left(\b0,\kappa_{\mathrm{UE}} P\Id_{\tau}\right) $ is the $ \tau \times 1$ additive transmit HWI vector while $ \bDelta_{\mathrm{r}} \in \mathbb{C}^{M \times \tau} $ is the additive receive HWI matrix where each column is distributed as $ \deltav_{\mathrm{r}} \in \mathbb{C}^{M \times 1} \sim \mathcal{CN}\left(\b0,\bm \Upsilon \right)$ with $ \bm \Upsilon =\kappa_{\mathrm{BS}} P\sum_{i=1}^{K} \Id_{M}\circ \bh_{i} \bh_{i}^{\H}$. Note that the phase noise is hidden inside the expression of $ \bh_{k} $. In addition, $ \bW^{\tr} \in \mathbb{C}^{M \times \tau} $ is the AWGN matrix at the BS with independent columns, each one distributed as $ \mathcal{CN}\left(\b0,\sigma^2\Id_{M}\right)$.
	
	The received training signal at the BS, given by \eqref{train1}, is multiplied by the transmitted training sequence from UE $ k $ to eliminate the interference caused by other UEs, and obtain
	\begin{align}
		\br_{k}=\bh_{k}+\sum_{i=1}^{K}\frac{\tilde{\delta}_{\mathrm{t},i}}{ \tau P}\bh_{i} +\frac{\tilde{\deltav}_{\mathrm{r}}+\bw_{k}}{ \tau P},\label{train2}
	\end{align}
	where $ \tilde{\delta}_{\mathrm{t},i}=\deltav_{\mathrm{t},i}^{\H}\bx_{\mathrm{p},k} $, $\tilde{\deltav}_{\mathrm{r}}=\bDelta_{\mathrm{r}}\bx_{\mathrm{p},k} $, and $ \bw_{k}=\bW^{\tr} \bx_{\mathrm{p},k}$. 

	\begin{proposition}\label{PropositionDirectChannel}
		The LMMSE estimate of the overall channel $ \bh_{k} $ is given by
		\begin{align}
			\hat{\bh}_{k}=\bR_{k}\bQ_{k} \br_{k},\label{estim1}
		\end{align}
		where $ \bQ_{k}\!=\! \left(\!\bR_{k}\!+ \!\frac{\kappa_{\mathrm{UE}}}{\tau }\!\sum_{i=1}^{K} \!\bR_{i}\!+\! \frac{\kappa_{\mathrm{BS}}}{\tau}\!\sum_{i=1}^{K}\!\Id_{M}\!\circ\! \bR_{i} \!+\!\frac{\sigma^2}{ \tau P }\Id_{M}\!\right)^{\!-1}$, and $ \br_{k}$ is the noisy observation of the effective overall channel from UE $k$ given by \eqref{train2}.
	\end{proposition}
	\begin{proof}
		The proof is provided in Appendix~\ref{Proposition1}.	
	\end{proof}
	
	According to the property of orthogonality of LMMSE estimation, the overall perfect channel is given by 
	\begin{align}
		\bh_{k}=\hat{\bh}_{k}+\tilde{\bh}_{k},\label{current} \end{align}
	where $\hat{\bh}_{k}$ and $\tilde{\bh}_{k} $ have zero mean and variances $	\bPsi_{k}\!=\!\bR_{k}\bQ_{k}\bR_{k}$ and $ \tilde{\bPsi}_{k}=\bR_{k}-\bPsi_{k}$, respectively.  Contrary to conventional estimation theory concerning independent Gaussian noise, $\hat{\bh}_{k}$ and $\tilde{\bh}_{k}$ are neither independent nor jointly complex Gaussian vectors because the effective distortion noises are not Gaussian, e.g., $ \bh_{k} \tilde{\delta}_{\mathrm{t},k}$ is the product between two Gaussian variables. However, $\hat{\bh}_{k}$ and $\tilde{\bh}_{k} $ are uncorrelated and each of them has zero mean~\cite{Bjoernson2017}. In the unrealistic case of perfect HWIs, the LMMSE estimator of the overall channel vector coincides with the optimal MMSE estimator. Notably, the CE can be easily generalized to include other fading models such as independent Rayleigh fading, where $ \bR_{\mathrm{IRS},k}=\Id_{N} $ and $ \bR_{\mathrm{BS},k}=\Id_{M} $.

	\begin{remark}\label{rem3}
		A comparison with other CE schemes is difficult since the majority of works such as \cite{He2019} does not yield analytical expressions, while the proposed method is indicated for future closed-form manipulations. Compared to \cite{Nadeem2020} requiring $ N+1 $ subphases, our proposed method has a lower training overhead requiring only one phase and has achieved to obtain the estimated cascaded channel vector while, therein, only the estimated channel (scalar) concerning each element was obtained.
	\end{remark}
	
	\begin{remark}\label{rem2}
		Generally, the covariances $ \bR_{k} $, $ \bPsi_{k} $, and  $ \tilde{\bPsi}_{k} $ depend on both the T-HWIs and the IRS-HWIs. In the special case of uniformly distributed phase errors, these covariances do not depend on these errors or the reflect phase matrix $ \bTheta $. Then, we can not take benefit from any IRS optimization to minimize the estimation error and achieve better estimation.
	\end{remark}
	
	The study of the NMSE is insightful. Specifically, we define
	\begin{align}
		\mathrm{NMSE}_{k}&=\frac{\tr(\EE[(\hat{\bh}_{k}-{\bh}_{k})(\hat{\bh}_{k}-{\bh}_{k})^{\H}])}{\tr\left(\EE[{\bh}_{k}{\bh}_{k}^{\H}]\right)}\\
		&=1-\frac{\tr(\bPsi_{k})}{\tr(\bR_{k})}.\label{nmse1}
	\end{align}
	
	The T-HWIs are found inside $ \bPsi_{k} $ in terms of $ \bQ_{k} $ while the IRS-HWIs (phase noise) appear inside $ \bR_{k} $. From \eqref{nmse1}, we observe that an increase of the T-HWIs results in the increase of the $ \mathrm{NMSE}_{k} $. Moreover, according to Remark \ref{rem1}, if the phase errors are uniformly distributed, the NMSE does not depend on the RBM and the NMSE can not be optimized. A similar observation takes place if uncorrelated Rayleigh fading is assumed.
	%
	\section{Uplink Data Transmission with HWIs}\label{PerformanceAnalysis}
	In this section. we focus on the derivation of the uplink achievable sum SE of a practical IRS-aided MU-MISO setup with HWIs. The received signal by the BS can be written as
	\begin{align}
		\by=\sum_{i=1}^{K}\bh_{i}\left(s_{i}+\delta_{\mathrm{t},i}\right) +\deltav_{\mathrm{r}}+\bn,\label{ULTrans}
	\end{align}
	where $\bn\sim \mathcal{CN}\left(\b0,\sigma^2\Id_{M}\right) $ is the AWGN at the BS and phase noises are found inside the expression of $ \bh_{i} $, while $ \delta_{\mathrm{t},i}$ and $ \deltav_{\mathrm{r}} $ correspond to the T-HWIs.
	\subsection{Achievable SE}	\label{lower1}
	The BS estimates $ s_{k} $ from UE $ k $ by means of~\eqref{ULTrans} in terms of linear single-user detection by applying the receive combining vector $ \bv_{k} $ as $ \hat{s}_{k}= \bv_{k}^{\H}\by_{}$. 
	Moreover,  we exploit the use-and-then-forget (UatF) bound, suggested for systems with a large number of antennas ( $ M>8 $) \cite{Bjoernson2017}, in order to obtain a closed-form expression of the SE. Note that this bound can be applied with different channel estimators (not only LMMSE) and decoders. Specifically, $ \hat{s}_{k} $ can be rewritten as
	\begin{align}
		\hat{s}_{k}&=\underbrace{\sqrt{\rho_{k}}\EE\left\{\bv_{k}^{\H}{\bh}_{k}\right\} s_{k}}_{\mathrm{DS}_{k}} + \underbrace{\sqrt{\rho_{k}}\left(\bv_{k}^{\H}{\bh}_{k}-\EE\left\{\bv_{k}^{\H}{\bh}_{k}\right\}\right) s_{k}}_{\mathrm{BU}_{k}}\nn\\
		&+\sum_{i\ne k}^{K}\underbrace{\sqrt{\rho_{i}}\bv_{k}^{\H}{\bh}_{i}s_{i}}_{\mathrm{MUI}_{ik}} +\sum_{i=1}^{K}\underbrace{\bv_{k}^{\H}{\bh}_{i}\delta_{\mathrm{i},k}}_{\mathrm{TD}_{i}}+\underbrace{\bv_{k}^{\H}\deltav_{\mathrm{r}}}_{ \mathrm{RD}_{k}}+\underbrace{\bv_{k}^{\H}}_{\mathrm{RN}_{k}},
	\end{align}
	where
	$ \mathrm{DS}_{k} $, $ \mathrm{BU}_{k} $, $ \mathrm{MUI}_{ik} $ express the desired signal (DS) part, the beamforming gain uncertainty (BU), and each term of the sum describing the MU interference (MUI). Also, $ \mathrm{TD}_{i} $, $ \mathrm{RD}_{k} $, and $ \mathrm{RN}_{k} $ express the transmit distortion, the receive distortion, and the receiver AWGN noise. Next, by applying a standard bound technique assuming worst-case uncorrelated additive noise for the inter-user interference and the distortion noise \cite{Hassibi2003}, we derive a lower bound on the uplink average SE in bps/Hz, which is known as the use-and-then-forget bound in the massive MIMO (mMIMO) literature \cite{Bjoernson2017}. In particular, the achievable SE is given by
	\begin{align}
		\mathrm{SE}_{k}	=\frac{\tau_{\mathrm{c}}-\tau}{\tau_{\mathrm{c}}}\log_{2}\left ( 1+\gamma_{k}\right)\!,\label{LowerBound}
	\end{align}
	where   the pre-log fraction expresses the percentage of samples per coherence block used for uplink data transmission and $ \gamma_{k}=\frac{S_{k}}{I_{k}}$ is the uplink SINR with 
	\begin{align}
		S_{k}&=|\mathrm{DS}_{k} |^{2},\label{sig11}\\
		I_{k}&=\EE\left\{|\mathrm{BU}_{k}|^{2}\right\}+\sum_{i\ne k}^{K}\EE\left\{\mathrm{MUI}_{ik}\right\}+\sum_{i=1}^{K} \EE\left\{|\mathrm{TD}_{i}|^{2}\right\}\nn\\
		& +\EE\left\{|\mathrm{RD}_{k}|^{2}\right\}\!+\!\EE\left\{|\mathrm{RN}_{k}|^{2}\right\}\!\label{int1}
	\end{align}
	describing the desired signal power and the interference plus noise power. 	For the sake of further convenience, we denote by $ \sigma_{\mathrm{UE}}^{2}= \sum_{i=1}^{K} \EE\left\{|\mathrm{TD}_{i}|^{2}\right\}$ and $ \sigma_{\mathrm{BS}}^{2}= \EE\left\{|\mathrm{RD}_{k}|^{2}\right\}$ the variances of the additive transmit and receive HWIs from all transmit UEs and at the output of the decoder.
	
	Generally, MRC and conventional MMSE decoders are the most common linear receivers for the uplink of next-generation systems such as mMIMO \cite{Hoydis2013,Papazafeiropoulos2015a}. However, the expectations in \eqref{sig11} and \eqref{int1} cannot be derived in closed-form in the case of the optimal MMSE receiver except if the deterministic equivalent analysis is  applied \cite{Hoydis2013,Papazafeiropoulos2015a,Papazafeiropoulos2016}. Also, the next step that includes the optimization with respect to reflection coefficients would be quite intractable. Hence, we focus on the derivation of a closed-form SINR by applying MRC decoding, which can be obtained even for a finite number of BS antennas.  Thus, below we assume $ \bv_{k}=\hat{\bh}_{k}$. However,  given the higher  performance expected by MMSE decoding, its application in the study of HWIs with statistical CSI according to the proposed methodology is the topic of ongoing work.

	\begin{Theorem}\label{theorem:ULDEMMSE}
		Given the RBM $ \bTheta $, the uplink achievable SINR of UE $k$ with MRC decoding in an IRS-assisted MU-MISO system, accounting for imperfect CSI and HWIs, is given by
		\begin{align}
			\gamma_{k}\!=\!	\frac{\bar{S}_{k}}{	\bar{I}_{k}},
2		\end{align}
		where
		\begin{align}
		&	\bar{S}_{k}\!=\!\rho_{k}|\!\tr\left(\bPsi_{k}\right)\!|^{2},\label{Num1}\\
		&	\bar{I}_{k}\!=\!\left(1\!+\!\kappa_{\mathrm{UE}}\right)\!\!\left(\sum_{i=1}^{K}\!\rho_{i}\!  \tr\!\left(\bPsi_{k} \bR_{i}\right)\!-\!\rho_{k}\!\tr\left( \bPsi_{k}^{2}\right)\!\!\!\right)\!\!+\!\rho_{k}\kappa_{\mathrm{UE}}|\!\tr\!\left(\bPsi_{k}\!\right)|^{2}\nn\\
			&+\!\bkappa_{\mathrm{BS}} \!\left(\!\!\rho_{k}|\! \tr \!\left( \Id_{M}\!\circ\! \bPsi_{k}\right)\!|^{2}\! +\!\!\sum_{i=1}^{K}\!\rho_{i}\!\tr \left(\left(\Id_{M}\!\circ\! \bR_{i}\right) \!\bPsi_{k}\right)\!\!\right)\!\!+\!\sigma^2\! \tr\!\left(\bPsi_{k}\right)\!.\label{Den1}
		\end{align}
	\end{Theorem} 
	\proof The proof  is provided in Appendix~\ref{theorem1}.\endproof
	\begin{remark}
		Theorem \ref{theorem:ULDEMMSE}	provides the uplink achievable SINR with MRC under imperfect CSI in closed-form. Notably, it shows the impact of the unavoidable HWIs. Especially, it depends directly on the T-HWIs by means of $ \bkappa_{\mathrm{BS}} $ and $ \bkappa_{\mathrm{UE}} $. The impact of the phase noise appears indirectly through the covariance matrices. Moreover, the expression of 	$ \gamma_{k} $ depends only on slowly-varying large-scale statistics.
	\end{remark}

	Based on $ \gamma_{k} $, provided by Theorem \ref{theorem:ULDEMMSE}, the system (sum) SE in bps/Hz is obtained as
	\begin{align}
		\mathcal{R}=\frac{\tau_{\mathrm{c}}-\tau}{\tau_{\mathrm{c}}}\sum_{i=1}^{K}\log_{2}\left(1+ \gamma_{ i}\right).\label{sumse}
	\end{align}
	\subsection{IRS Design Problem: Formulation and Solution}\label{IRSdesign}
	IRS-aided architectures require to design the corresponding RBM, found inside the covariance matrices, in order to maximize the sum SE given by \eqref{sumse}. Hence, by resorting to the common assumption of infinite resolution phase shifters, herein, we formulate and solve 	the RB design problem under MRC and realistic conditions accounting for imperfect CSI and HWIs as
	\begin{align}\begin{split}
			(\mathcal{P}1)~~~~~~~\max_{\bTheta} ~~~	&\mathcal{R}\\
			\mathrm{s.t}~~~&|\phi_{n}|=1,~~ n=1,\dots,N,
		\end{split}\label{Maximization} 
	\end{align}
	with $ 	\mathcal{R} $ given by \eqref{sumse} and $ \phi_{n}= \exp\left(j \theta_{n}\right) $ are the elements of $ \bTheta $. Obviously, $ 	(\mathcal{P}1) $ is a non-convex maximization problem with respect to $ \bTheta $ with a unit-modulus constraint regarding $ \phi_{n} $. 	
	
	\begin{remark}\label{rem5}
		If the phase noise is uniformly distributed, the covariance matrices will not include the RBM according to Remark \ref{rem2}. Hence, the SINR/SE cannot be optimized, and the IRS does not serve its purpose.
	\end{remark}
	\begin{remark}
		Given that HWIs at the transceiver and IRS degrade the performance, the use of cheaper (lower quality) hardware will have a direct impact on the QoS. In such a case, a better RB design is suggested to compensate for the loss and improve the performance.
	\end{remark} 
	

	Taking the expression of $ \gamma_{k} $ into account, the optimization problem takes the form of a constrained maximization problem with a solution given by means of projected gradient 	ascent until converging to a 	stationary point as in \cite{Kammoun2020}. At every step, we project the solution onto the closest feasible point satisfying the unit-modulus constraint concerning $ \phi_{n} $. In more detail, the procedure assumes the vectors $ \bs^{i} =[\phi_{1}^{i}, \ldots, \phi_{N}^{i}]^{\T}$ including 	the induced phases at step $ i $. The next iteration point, resulting in the increase of $ \mathcal{R} $ towards to its convergence, is given by
	\begin{align}
		\tilde{\bs}^{i+1}&=\bs^{i}+\mu \bq^{i},\label{sol1}\\
		\bs^{i+1}&=\exp\left(j \arg \left(\tilde{\bs}^{i+1}\right)\right),\label{sol2}
	\end{align}
	where $ \mu $ is the step size and $ \bq^{i} $ describes the adopted ascent direction at step $ i $. In particular, we have $ [\bq^{i}]_{n}= \pdv{	\mathcal{R}}{\phi_{n}^{*}} $, which is obtained by Proposition \ref{Prop:optimPhase} below. 	 The suitable step size is computed at each iteration by means of the backtracking line search \cite{Boyd2004}. The solution of the problem, described by \eqref{sol1} and \eqref{sol2}, is found based on the projection problem $ \min_{|\phi_{n} |=1, n=1,\ldots,N}\|\bs-\tilde{\bs}\|^{2} $ under the unit-modulus constraint. The outline of the algorithm is described by Algorithm \ref{Algoa1}.
	\begin{algorithm}
		\caption{Projected Gradient Ascent Algorithm for the IRS Design}
		1.				 \textbf{Initialisation}: $ \bs^{0} =\exp\left(j\pi/2\right)\one_{N}$, $ \bTheta^{0}=\diag\left(\bs^{0}\right) $, $ 	\mathcal{R}^{0}=f\left(\bTheta^{0}\right) $ given by \eqref{sumse}; $ \epsilon>0 $\\
		2. \textbf{Iteration} $ i $: \textbf{for} $ i=0,1,\dots, $ do\\
		3. $[\bq^{i}]_{n}= \pdv{	\mathcal{R}}{\phi_{n}^{*}}, n=1, \ldots,N $, where $\pdv{	\mathcal{R}}{\phi_{n}^{*}} $ is given by Proposition \ref{Prop:optimPhase};\\
		4. \textbf{Find} $ \mu $ by backtrack line search$( f\left(\bTheta^{0}\right),\bq^{i},\bs^{i})$ \cite{Boyd2004};\\
		5. $ \tilde{\bs}^{i+1}=\bs^{i}+\mu \bq^{i} $;\\
		6. 	$ \bs^{i+1}=\exp\left(j \arg \left(\tilde{\bs}^{i+1}\right)\right) $; $ \bTheta^{i+1}=\al \diag\left(\bs^{i+1}\right) $;\\
		7. $ \mathcal{R}^{i+1}=f\left(\bTheta^{i+1}\right) $;\\
		8. \textbf{Until} $ \| \mathcal{R}^{i+1}- \mathcal{R}^{i}\|^{2} <\epsilon$; \textbf{Obtain} $ \bTheta^{*}=\bTheta^{i+1}$;\\
		9. \textbf{end for}\label{Algoa1}
	\end{algorithm}
	
	 The convergence of the proposed algorithm to a local maximum can be guaranteed because it is bounded due to the power constraint and it increases by setting $ [\bq^{i}]_{n}= \pdv{ \mathcal{R}}{\phi_{n}^{*}} $, where the backtracking line search is used to find a suitable step size.
	
	\begin{proposition}\label{Prop:optimPhase}
		The derivative of $ \mathcal{R} $ with respect to $ \phi_{n} $ is provided by
		\begin{align}
			\pdv{	\mathcal{R}}{\phi_{n}}=\sum_{k=1}^{K} \frac{\pdv{\bar{S}_{k}}{\phi^{*}_{n}}\bar{I}_{k}-\bar{S}_{k}\pdv{\bar{I}_{k}}{\phi^{*}_{n}}}{\ln (2) \bar{I}_{k}^{2}\left(1+\frac{\bar{S}_{k}}{\bar{I}_{k}}\right)},\label{der10}
		\end{align}
		where 	$ \bar{S}_{k} $, $ \bar{I}_{k} $ follow by \eqref{Num1}, \eqref{Den1} while
		\begin{align}
			&\pdv{\bar{S}_{k}}{\phi^{*}_{n}}	=2\rho_{k}\tr\left(\bPsi_{k}\right)\bL\! \left(\bR_{k},\bR_{k},\Id_{M} \right),\\
			&\pdv{\bar{I}_{k}}{\phi^{*}_{n}}=\left(1+\kappa_{\mathrm{UE}}\right)\big(\sum_{i=1}^{K}\!\rho_{i}\big(\bL \big(\bR_{i}\bR_{k},\bR_{i}\bR_{k},\bR_{i} \big)	\nn\\&\!\!+\!\al\beta_{2,k}[\bH_{1}^{\H}\bPsi_{k}\bH_{1} \bTheta\bR_{\mathrm{IRS},k}]_{n,n}\big)\nn\\
			&\!\!-\!2\rho_{k}\bL \big(\bPsi_{k}\bR_{k},\bPsi_{k}\bR_{k},\bPsi_{k}\big)\big)\nn\\
			&\!\!+\!\bkappa_{\mathrm{BS}} \big(2\rho_{k} \tr \big( \Id_{M}\circ \bPsi_{k}\big) \bL \big(\bR_{k},\bR_{k}, \Id_{M}\big)\!\nn\\ &\!\!+\!\sum_{i=1}^{K}\!\rho_{i}\big(\al\beta_{2,k}[\bH_{1}^{\H}\bPsi_{k}\bH_{1} \bTheta\bR_{\mathrm{IRS},k}]_{n,n}\nn\\
			&			\!\!+\!\bL \big(\big(\Id_{M}\circ \bR_{i}\big)\bR_{k},\big(\Id_{M}\circ \bR_{i}\big)\bR_{k},\Id_{M}\circ \bR_{i}\big)\big)\big)\nn\\
			&\!\!+\!2\rho_{k}\kappa_{\mathrm{UE}}\tr\big(\bPsi_{k}\big)\bL \big(\bR_{k},\bR_{k},\Id_{M} \big)\!+\!\sigma^2 \bL \big(\bR_{k},\bR_{k},\Id_{M} \big)
		\end{align}
		with $ 	\bL\! \left(\bA,\bB, \bC\right) $		 given by \eqref{LABC}									
		for any $ \bA\in\mathbb{C}^{N \times N} $, $ \bB \in\mathbb{C}^{N \times N} $, and $ \bC\in \mathbb{C}^{N \times N} $.	\end{proposition}
	\proof The proof of Proposition~\ref{Prop:optimPhase} is given in Appendix~\ref{optimPhase}.\endproof
	
	  The RBM  beamforming design, based on the  gradient ascent,  results in an outstanding
		performance since  the  gradient ascent is obtained in a closed-from  with low computational complexity based on simple matrix operations. In particular, the complexity of \eqref{der10} is $ \mathcal{O}\left(MN^{2}+M^{2}N+M^{2}K\right) $. Obviously, it depends on all fundamental system parameters, i.e., $ K $, $ M $, and $N  $  but with a higher (square) dependence on $ M $ and $ N $.
	\begin{remark}
		 If we do not have a closed-form expression for the SE,  we cannot apply the proposed method. Also, although the proposed algorithm, given by \eqref{sol1} and \eqref{sol2}, does not provide a global optimum but a locally optimal solution due to the non-convexity of the initial optimization problem with respect to the phase shifts, it offers a good preliminary tool to study IRS-aided systems under realistic conditions in terms of imperfect CSI and HWIs.
	\end{remark}
	\begin{remark}
		The dependence of the algorithm on the large-scale channel statistics achieves a reduction in the signal exchange overhead between the IRS controller and the BS since it will take place every several coherence intervals defined by the variation of these statistics.  On the contrary, on models, relying on the instantaneous CSI, the optimization should take place at every coherence interval, which results in large overhead, especially, when the IRS is large. We highlight that the proposed method can be exploited in both low-speed and fast-speed scenarios. The only difference is that in fast-speed scenarios, the large-scale statistics  change faster. For this reason, in such cases, the optimization should take place more frequently. Also, the simple expression of $ \pdv{ \mathcal{R}}{\phi_{n}^{*}} $ results in a significant decrease of the computational complexity.  These two reasons make the proposed method quite beneficial.
	\end{remark}
	\begin{remark}
		The proposed methodology for the RB design, i.e., the optimization of the phase matrix after having derived the performance expression in terms of large-scale statistics, could also be used in the case of negligence of HWIs, where another property of an IRS-assisted system would be the main topic of study. In such a case, the optimization in terms of the derivative will be simplified even more since the additional terms, concerning the T-HWIs which depend on the overall channel (the optimization variables $ \phi_{n} $), will be omitted. 
	\end{remark}

	\section{Numerical Results}\label{Numerical} 
	In this section, we depict and discuss the analytical results corresponding to the uplink performance in terms of CE and achievable sum SE of an IRS-aided MU-MISO system with imperfect CSI and HWIs. Monte-Carlo (MC) simulations ($ 10^{3} $ independent channel realizations) represented by "\ding{53}" marks in Figs. \ref{Fig0} and \ref{Fig2} below, corroborate our analysis and the tightness of the UatF bound. For the sake of comparison, we have modified \cite{Nadeem2020} to describe the uplink transmission. 
	\subsection{Simulation Setup}
	We consider a uniform linear array (ULA) of $ M $ antennas ($ M=16 $) at the BS, assisted by an IRS with a uniform planar array (UPA) of $ N $ elements ($ N=60 $) that serve $ K=5 $ UEs. The spatial correlation coefficient for the IRS between elements $ n $ and $ n' $ corresponding to UE $ k $ is given by \cite{Kammoun2020}.
\!\!	\begin{align}
		[\bR_{\mathrm{IRS},k}]_{n,n'}\!=\!\EE\big[\!\exp\!\big(j \frac{2 \pi}{\lambda}d_{\mathrm{IRS}}\!\left(n-n'\right)\sin \phi_{k}\sin \theta_{k}\big)\!\big],\label{RIS}
	\end{align}
	where $ d_{\mathrm{IRS}}=0.5\lambda $ while $ \phi_{k} $ and $ \theta_{k} $ express the elevation and azimuth angles for UE $ k $, and are generated by the Laplace and the Von Mises distribution, respectively. In the former case, we assume that the mean angle of departure and the spread are $ 90^{\circ} $ and $ 8^{\circ} $, respectively. The latter distribution is generated with mean angle of departure $ 0 $ and spread $ 0.2 $. Regarding the parameters for the channel matrix $ \bH_{1} $ between the BS and the IRS, we assume $ d_{\mathrm{BS}} = 0.5\lambda $ and $ \theta_{1,n} $, $ \psi_{1,n} $ are uniformly distributed between $ 0 $ to $ \pi $ and $ 0 $ to $ 2\pi $, respectively. Also, $ \theta_{2,n}= \pi- \theta_{1,n} $, $ \psi_{2,n}=\pi+ \psi_{1,n}$. The correlation matrix $ \bR_{\mathrm{BS},k} $ is obtained similar to \cite{Hoydis2013}. Moreover, the overall path loss for the IRS-assisted link is given by \cite{Wu2019a,Bjoernson2019b,Kammoun2020}
	\begin{align}
		\bar{\beta}_{k}=C_{1}C_{2} d_{\mathrm{BS}-\mathrm{IRS}}^{-\al_{1}}d_{\mathrm{IRS}-\mathrm{UE}_{k}}^{-\al_{2}},
	\end{align}
	where $ \beta_{1}=\frac{C_{1}}{d_{\mathrm{BS}-\mathrm{IRS}}^{\al_{1}}} $ and $ \beta_{2,k}=\frac{C_{2}}{d_{\mathrm{IRS}-\mathrm{UE}_{k}}^{\al_{2}}} $ are the channel attenuation coefficients between the BS and the IRS and between 	the IRS and UE $ k $, respectively. Note that $ \al_{1} $ and $ d_{\mathrm{BS}-\mathrm{IRS}} $ are the path-loss exponent and distance concerning the link BS-to-IRS, while $ \al_{2} $ and $ d_{\mathrm{IRS}-\mathrm{UE}_{k}} $ are the path-loss exponent and distance concerning the link IRS-to-UE $ k $.
	The parameters values are chosen relied on the 3GPP Urban Micro (UMi) scenario from TR36.814 for a carrier frequency of $ 2.5 $ GHz and noise level $ -80 $ dBm, where the path losses for $ \bh_{2,k} $ and $ \bH_{1} $ are generated based on the NLOS and LOS versions \cite{3GPP2017}. Hence, we have $ \al_{1} =2.2$ and $ \al_{2} =3.67$ while $ d_{\mathrm{BS}-\mathrm{IRS}} =8~\mathrm{m}$ and $ d_{\mathrm{IRS}-\mathrm{UE}_{k}} =60~\mathrm{m}$. Also, $ C_{1}=26 $ dB, $ C_{2}=28 $ dB by assuming $ 5 $ dBi antennas at the BS and IRS while each UE includes a single
	0dBi antenna \cite{Bjoernson2019b}. For $ \beta_{\mathrm{d},k} $, we assume the same parameters as for $ \beta_{2,k} $, but we also consider an additional penetration loss equal to $ 15~\mathrm{dB} $. In addition, we assume that the coherence bandwidth is $B_{\mathrm{c}}= 200~\mathrm{KHz} $ and the coherence time is $ T_{\mathrm{c}}=1~\mathrm{ms} $, i.e., each coherence block consists of $\tau_{\mathrm{c}}=200$ samples. Also, we assume $ P=6~\mathrm{dB}$ and the same value for $ p_{i}=\rho,~\forall i $ during the uplink data transmission. Note that $ \sigma^2=-174+10\log_{10}B_{\mathrm{c}} $.
	
	For the study of the T-HWIs, we assume that we have an Analog-to-Digital Converter (ADC) at the BS, which quantizes the received signal to a $b$ bit resolution. As a consequence, the receive distortion can be written as $\kappa_{\mathrm{BS}}=2^{-2b}/\left(1-2^{-2b}\right)$, which gives $ \kappa_{\mathrm{BS}}= 0.258^{2},~0.126^{2}$, and $0.062^{2}$ for $ b=2, 3 $, and $ 4 $ bits, respectively ~\cite{Bjornson2015,Papazafeiropoulos2017}. In other words, a smaller resolution results in more severe distortion. The same value is used for $ \kappa_{\mathrm{UE}} $. Note that the trend in 5G networks and beyond is the use of lower precision ADCs. Regarding, the IRS-HWIs, if the Von Mises PDF is assumed to model the phase noise, the concentration parameter is set to $ \kappa_{\tilde{\theta}}=2 $. Unless otherwise stated, this set of parameters is used during the simulations.
	
	\begin{figure}[!h]
		\begin{center}
			\includegraphics[width=0.95\linewidth]{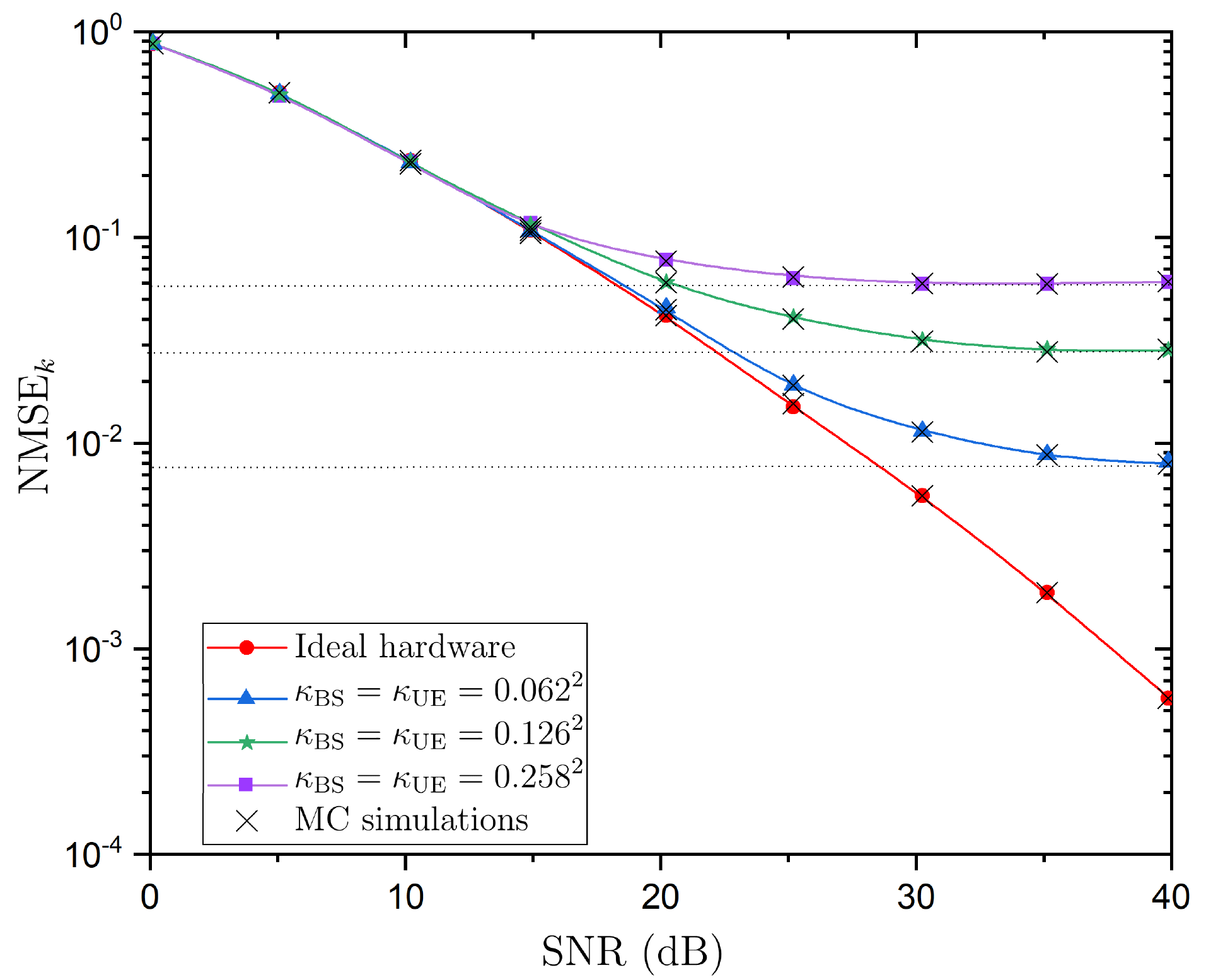}
			\caption{\footnotesize{NMSE of UE $ k $ versus the SNR of an IRS-assisted MIMO system with imperfect CSI ($ M=16 $, $ N=60 $, $ K=5 $) for varying T-HWIs $ \kappa_{\mathrm{BS}} $, $\kappa_{\mathrm{UE}}$ in the cases of uniform PDF for the phase noise or uncorrelated fading at the IRS (Analytical results and MC simulations). }}
			\label{Fig0}
		\end{center}
	\end{figure}
	Fig. \ref{Fig0} illustrates the relative estimation error per channel element, i.e., the normalized mean square error (NMSE) with respect to the uplink SNR for different values of the T-HWIs defining certain noise floors (asymptotic limits as $ p \to \infty $). In addition, we show the result corresponding to perfect hardware. Obviously, this line decreases without bound. Moreover, it is shown that the error floors go higher with increasing the severity of T-HWIs. Even at mild values of T-HWIs, we observe that the NMSE approaches the corresponding floor after $ 20~\mathrm{dB}$, which means that IRS-assisted systems require a high SNR to operate since they are dependent on a conventional MU-MISO architecture. For the sake of exposition, we have assumed uniform phase noise or uncorrelated fading at the IRS, in order to avoid any RB optimization since the NMSE does not depend on $ \bTheta $ in these cases (Remark~\ref{rem1}). Optimization with respect to $ \bTheta $ has been performed only in the case of the sum SE. Based on Remark \ref{rem3}, comparisons with other methods cannot be made at this stage. However, a comparison with respect to \cite{Nadeem2020} takes place below (Fig. \ref{Fig1}) in the case of achievable sum SE. Notably, MC simulations verify the analytical results.
	\begin{figure}[!h]
		\begin{center}
			\includegraphics[width=0.95\linewidth]{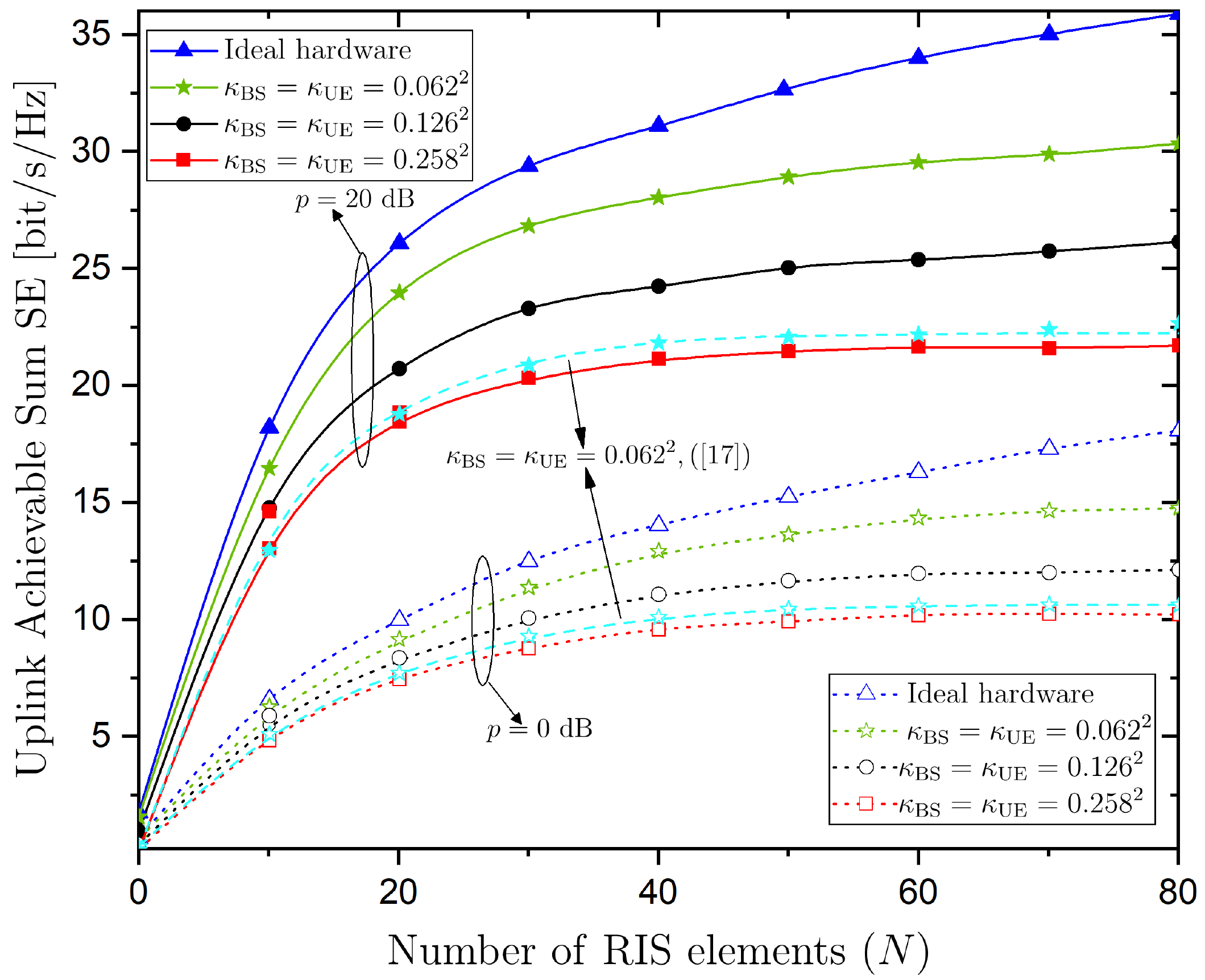}
			\caption{\footnotesize{Uplink achievable sum SE versus the number of IRS elements $N$ of an IRS-assisted MIMO system with imperfect CSI ($ M=16 $, $ K=5 $) for varying T-HWIs $ \kappa_{\mathrm{BS}} $, $\kappa_{\mathrm{UE}}$ and transmit power $ p $. }}
			\label{Fig1}
		\end{center}
	\end{figure}
	
	In Fig. \ref{Fig1}, we depict the achievable sum SE versus the number of IRS elements $ N $ for two different SNR values, $ \rho=0~\mathrm{dB}$ and $ \rho=20~\mathrm{dB}$. Also, we have considered different values of T-HWIs. First, we observe an increase of $ \mathcal{R} $ with $ N $, which increases unboundedly in the case of ideal hardware, but saturates for imperfect hardware met in practice. As expected, the degradation is higher when T-HWIs are more severe, probably, in the case of cheaper hardware used for a cost-efficient implementation. Furthermore, the convergence speed is faster at the higher SNR group because the T-HWIs are power-dependent. Thus, the largest part of the gain is achieved at lower values of $ N $. However, an increase of the IRS elements still allows for a further increase of $ \mathcal{R} $. For the sake of comparison, we have considered the CE from \cite{Nadeem2020} in terms of simulation ("dashed-star" lines) when $ \kappa_{\mathrm{BS}}=\kappa_{\mathrm{UE}}=0.258^{2} $. The achievable sum SE is much lower than the proposed method because of the high training overhead.
	\begin{figure}[!h]
		\begin{center}
			\includegraphics[width=0.95\linewidth]{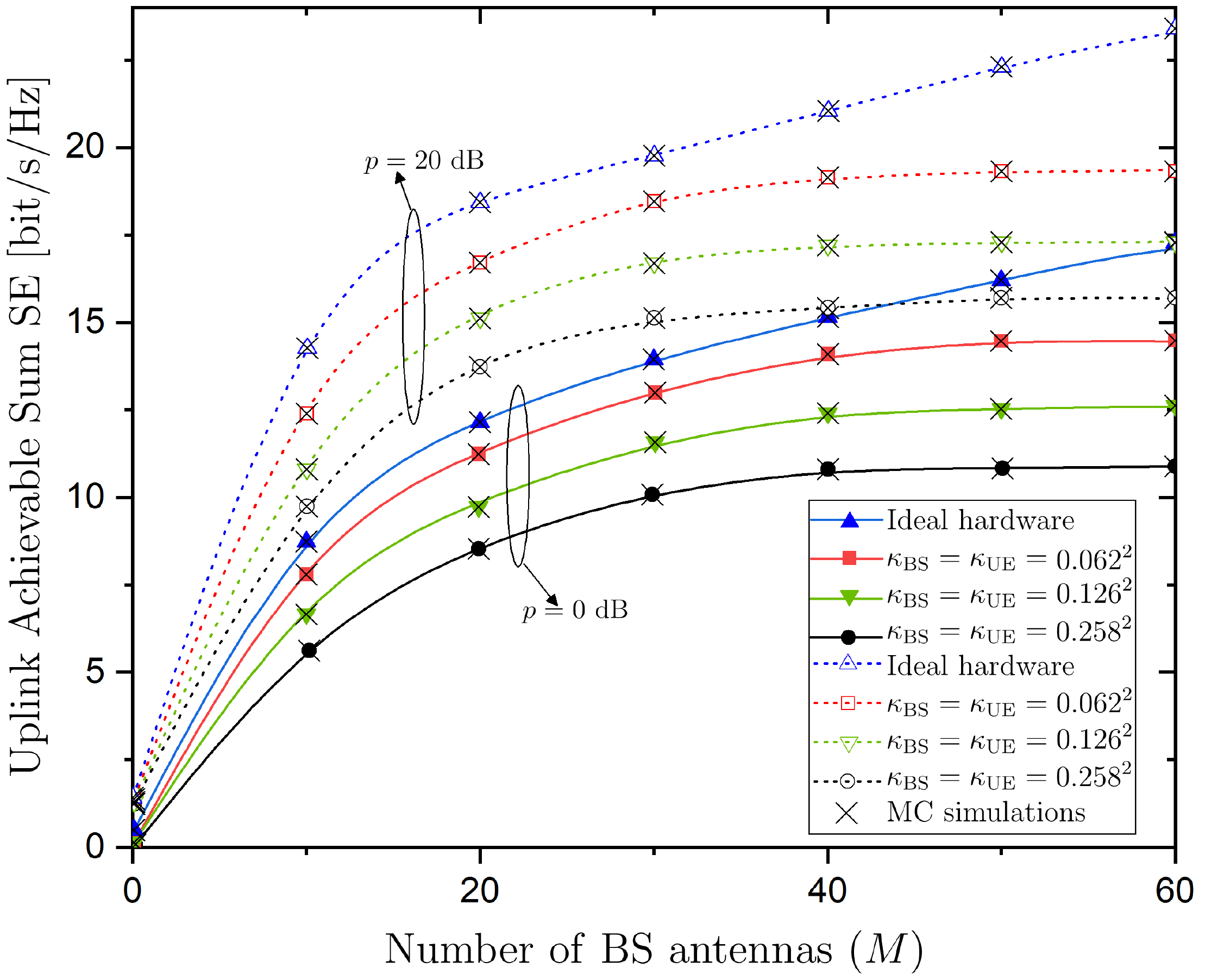}
			\caption{\footnotesize{Uplink achievable sum SE versus the number of BS antennas $M$ of an IRS-assisted MIMO system with imperfect CSI ($ N=60 $, $ K=5 $) for varying T-HWIs $ \kappa_{\mathrm{BS}} $, $\kappa_{\mathrm{UE}}$ and transmit power $ p $ (Analytical results and MC simulations). }}
			\label{Fig2}
		\end{center}
	\end{figure}
	
	Fig. \ref{Fig2} illustrates the achievable sum SE with respect to the number of BS antennas $ M $ for varying T-HWIs and SNR. Notably, this figure resembles with the previous figure, i.e., $ \mathcal{R} $ presents a similar dependence on $ M $ and $ N $. Hence, when $ M $ grows large, the sum SE increases without limit, if perfect hardware is assumed while it appears ceilings in practice, where T-HWIs exist. In fact, lower hardware quality results in larger degradation. In addition, by increasing the SNR, the sum SE becomes larger. Also, the sum SE saturates faster in the instance of a larger SNR ($ 20~\mathrm{dB} $). Thus, these two figures indicate that an IRS-assisted system performs better at higher SNR values, and with larger values of IRS elements and BS antennas. Note that the latter is further appealing since it agrees with the massive MIMO technology of which the implementation has already started. Moreover, the analytical results are accompanied by MC simulations showing the tightness and correctness of the lower bounds since they coincide. 
	\begin{figure}[!h]
		\begin{center}
			\includegraphics[width=0.95\linewidth]{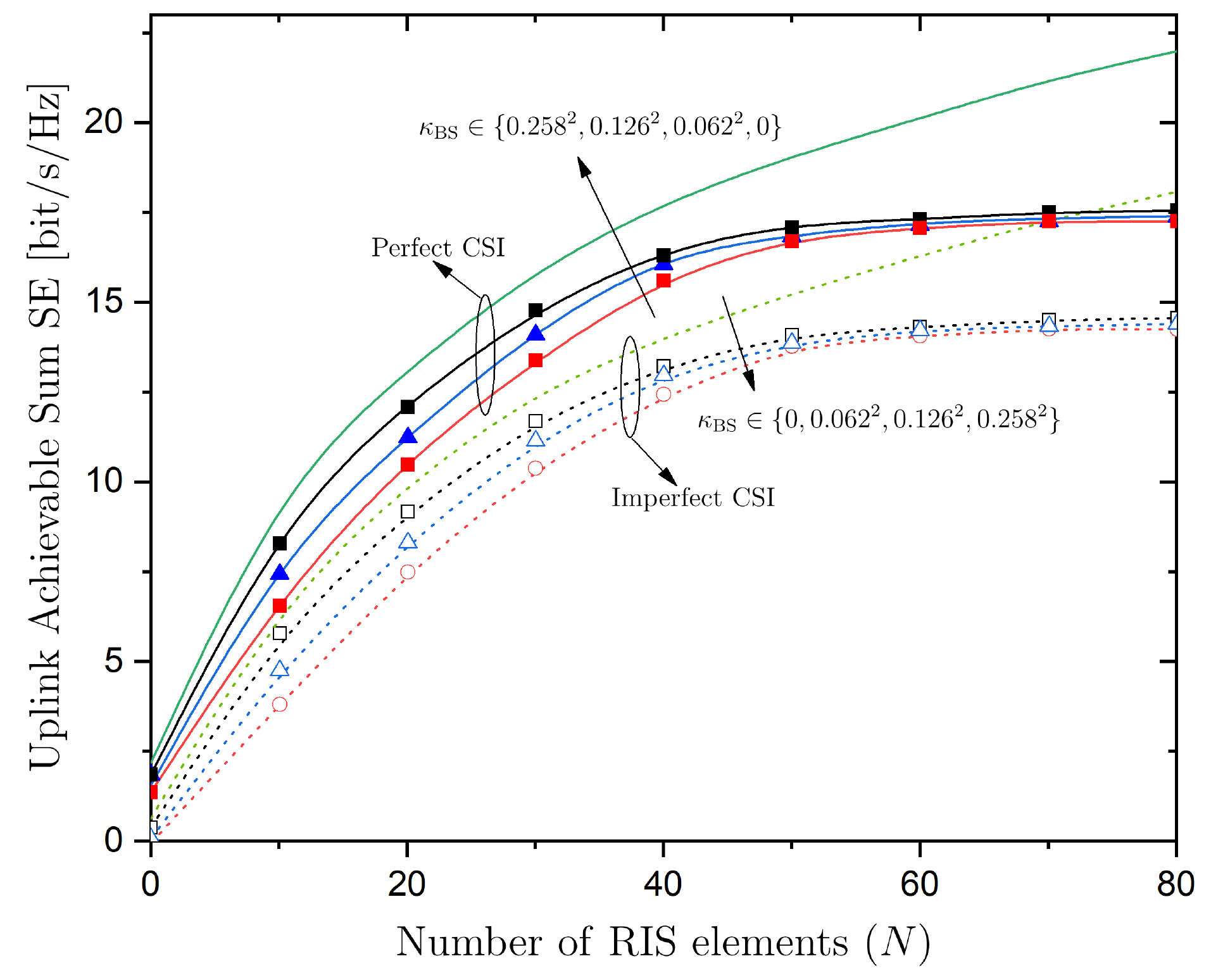}
			\caption{\footnotesize{Uplink achievable sum SE versus the number of IRS elements $N$ of an IRS-assisted MIMO system ($ M=16 $, $ K=5 $) for varying BS distortion $ \kappa_{\mathrm{BS}} $ and both cases of perfect and imperfect CSI. }}
			\label{Fig3}
		\end{center}
	\end{figure}
	
	Fig. \ref{Fig3} shows the achievable sum SE with respect to the number of IRS elements by varying the impact of the BS distortion $ \kappa_{\mathrm{BS}} $ while the distortion at the UE side is assumed zero. Notably, the lines converge to the same finite limit as $ N $ increases, which means that the impact of $ \kappa_{\mathrm{BS}}\ $ becomes negligible when $ N $ increases. In other words, despite the unavoidable existence of imperfect hardware, the IRS is suggested as $ N $ increases. Actually, it allows the use of low-cost hardware. Moreover, in the same figure, we have provided a comparison between perfect and imperfect CSI. We notice that, in both cases, the sum SE increases with the number of IRS elements while, again, the impact of $ \kappa_{\mathrm{BS}} $ becomes negligible at large $ N $. Obviously, in the realistic case of imperfect CSI, the achievable sum SE is lower, and the gap between the lines of ideal and imperfect hardware increases with larger $ N $ because the impact from the estimation error becomes larger.

	\begin{figure}[!h]
		\begin{center}
			\includegraphics[width=0.95\linewidth]{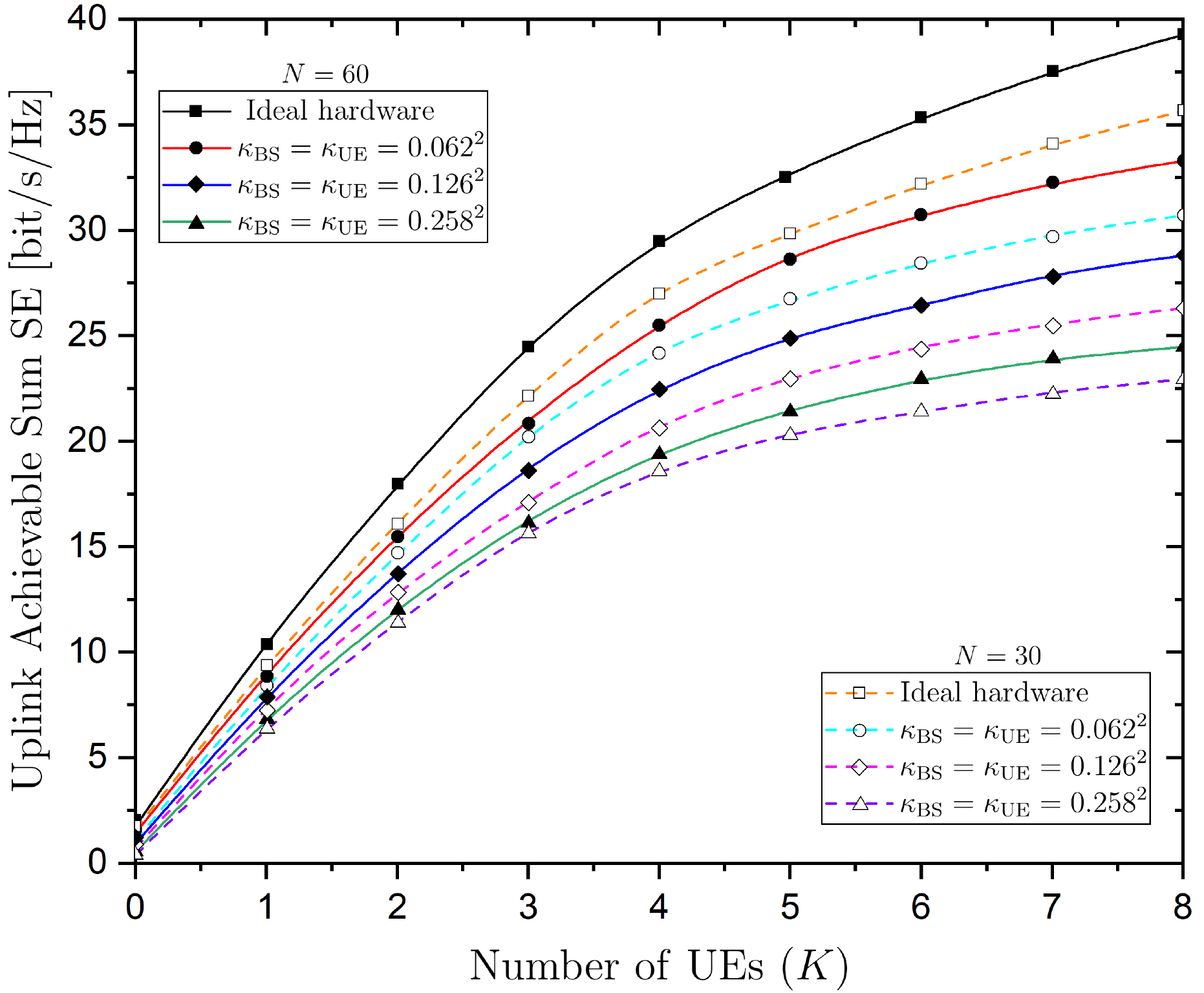}
			\caption{\footnotesize{ Uplink achievable sum SE versus the number of UEs $K$ of an IRS-assisted MIMO system ($ M=16 $,  $ \rho=20~\mathrm{dB}$) for varying T-HWIs $ \kappa_{\mathrm{BS}} $, $\kappa_{\mathrm{UE}}$ and IRS elements  $ N $.}}
			\label{Fig9}
		\end{center}
	\end{figure}

	 Fig. \ref{Fig9}  examines  the achievable sum SE with respect to the number of UEs  by varying the T-HWIs for different numbers of IRS elements.  The sum SE increases with $ K $ almost linearly at the beginning but the gradient decreases as  $ K $ increases. This result is reasonable since the increase of $ K $ increases the multi-user interference and the  received distortion as described by \eqref{eta_rU}. Also, a larger IRS in terms of $ N $ results in a larger sum SE as has been already shown.

	\begin{figure}[!h]
		\begin{center}
			\includegraphics[width=0.95\linewidth]{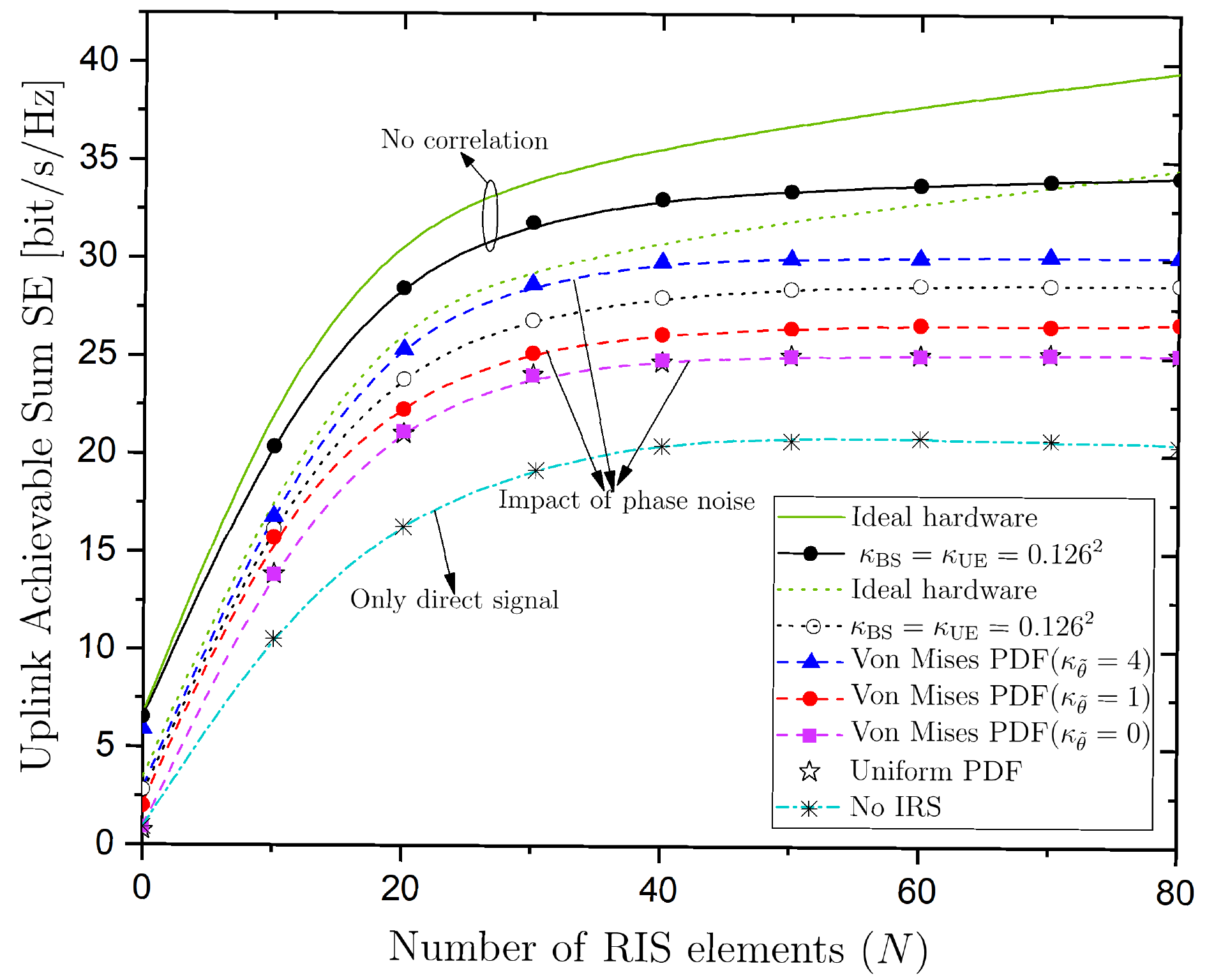}
			\caption{\footnotesize{Uplink achievable sum SE versus the number of IRS elements $N$ of an IRS-assisted MIMO system with imperfect CSI ($ M=16 $, $ K=5 $) for varying IRS-HWIs and T-HWIs in the cases of correlated/non-correlated Rayleigh fading. }}
			\label{Fig4}
		\end{center}
	\end{figure}
	In Fig. \ref{Fig4}, we show the impact of correlated Rayleigh fading on the achievable sum SE in the cases of correlation at both the BS and the IRS while varying the number of IRS elements. Also, we vary the quality of T-HWIs, and we show how for the same T-HWIs at the BS and UE sides, the sum SE decreases in the case of correlated fading ("dot" lines) with comparison to no correlation ("solid" lines). Furthermore, this figure allows shedding light on the impact of phase noise at the IRS. In particular, when uniform phase noise is assumed ("star" symbols), $ \mathcal{R} $ is the lowest because the design cannot take benefit (random fluctuations) from the IRS optimization since $ \bR_k $ does not depend on the phase matrix (see Rem. \ref{rem5}). However, in the case that the phase noise is distributed according to the Von Mises distribution ("dashed" lines), the achievable sum SE increases because the presence of the IRS becomes advantageous since it can adjust the phase shifts of its passive elements towards better performance. Especially, we have considered variation of the concentration parameter $ \kappa_{\tilde{\theta}} $. As $ \kappa_{\tilde{\theta}} $ decreases, the achievable sum SE decreases. Actually, we show that when $ \kappa_{\tilde{\theta}} =0$, the corresponding line coincides with the line describing the uniform PDF since the Von Mises PDF coincides with the uniform distribution in this case.  Note that we have also depicted the performance  in the absence of the IRS.   From the figure, we can verify our observation in Rem. \ref{rem1} explaining that the IRS contributes to the performance even in the worst-case IRS phase noise scenario since the line describing the case with no IRS is lower, i.e., the performance is worse.
	
	\begin{figure}[!h]
		\begin{center}
			\includegraphics[width=0.95\linewidth]{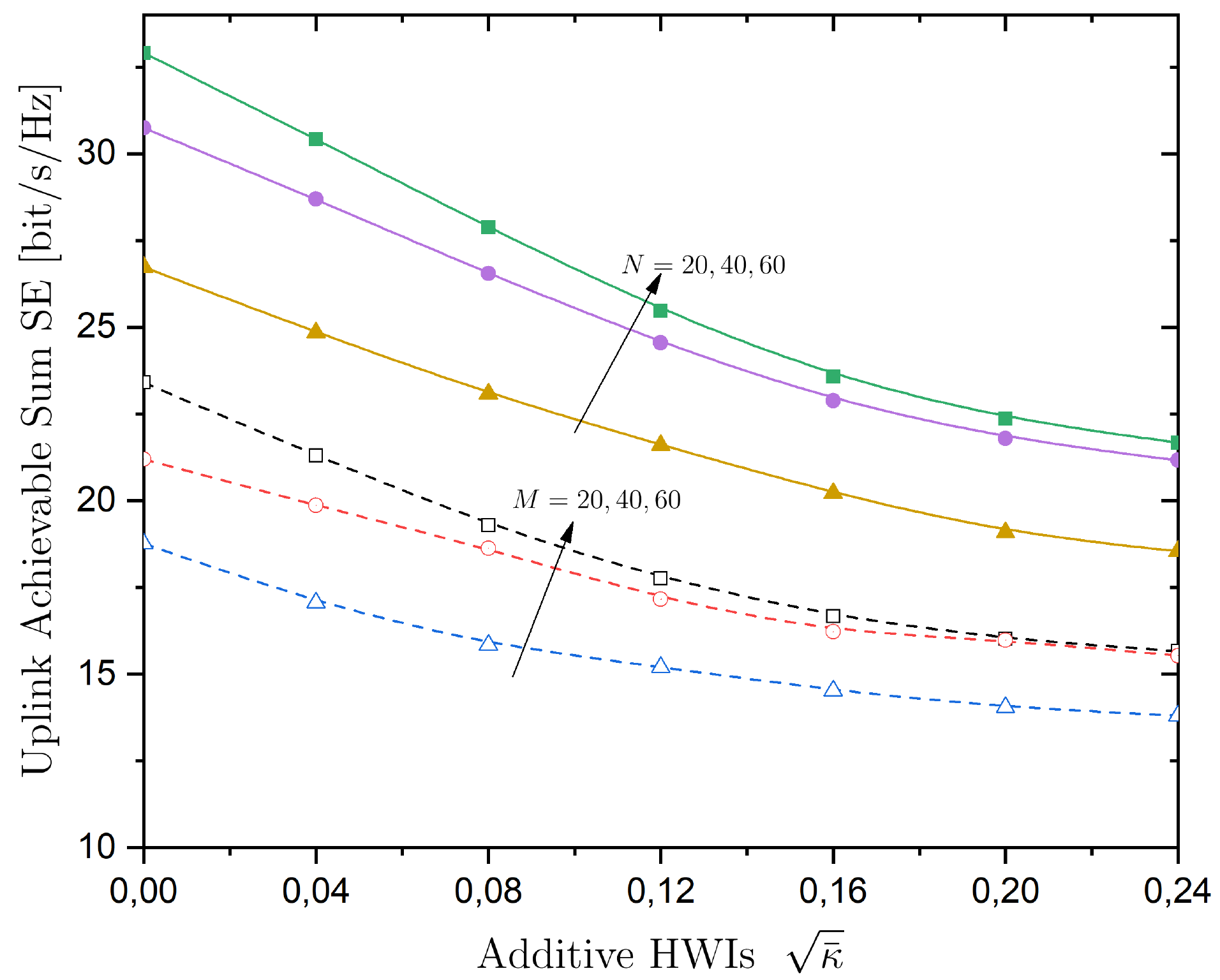}
			\caption{\footnotesize{ Uplink achievable sum SE versus the T-HWIs quality $ \sqrt[]{\bar{\kappa}} $   ($ \rho=20~\mathrm{dB}$, $ K=5 $) for varying BS antennas $ M $ and IRS elements $ N $. }}
			\label{Fig8}
		\end{center}
	\end{figure}

	 	Fig.~\ref{Fig8} shows  the performance of the achievable sum SE versus $ \sqrt[]{\bar{\kappa}} $, where $ \kappa_{\mathrm{BS}}=\bar{\kappa}  $ and $ \kappa_{\mathrm{UE}}=\bar{\kappa} +0.03 $. We have assumed that the distortion at the UE is larger due to its simpler hardware. The impact of IRS phase noise  is not considered as we  focus on  the impact of the T-HWIs.  The "solid" and "dashed" lines correspond to the variation of $ M $ BS antennas and $ N $ IRS elements, respectively. The horizontal axis starts from the case of no T-HWIs  at the BS, i.e., when $ \bar{\kappa}=0 $ and ends with severe additive HWIs. We observe that as T-HWIs increase, the performance decreases. Moreover, we observe that at severe T-HWIs the variations of the $ M $ and $ N $ do not affect the performance. Also, we notice that   the number of  IRS elements $ N $ affects more   the performance (higher sum SE) than the number of BS antennas $ M $ while the same variation with respect to $ N $ has a larger impact since the gaps between the solid lines are larger.

	\begin{figure}%
		\centering
		\subfigure[]{	\includegraphics[width=0.95\linewidth]{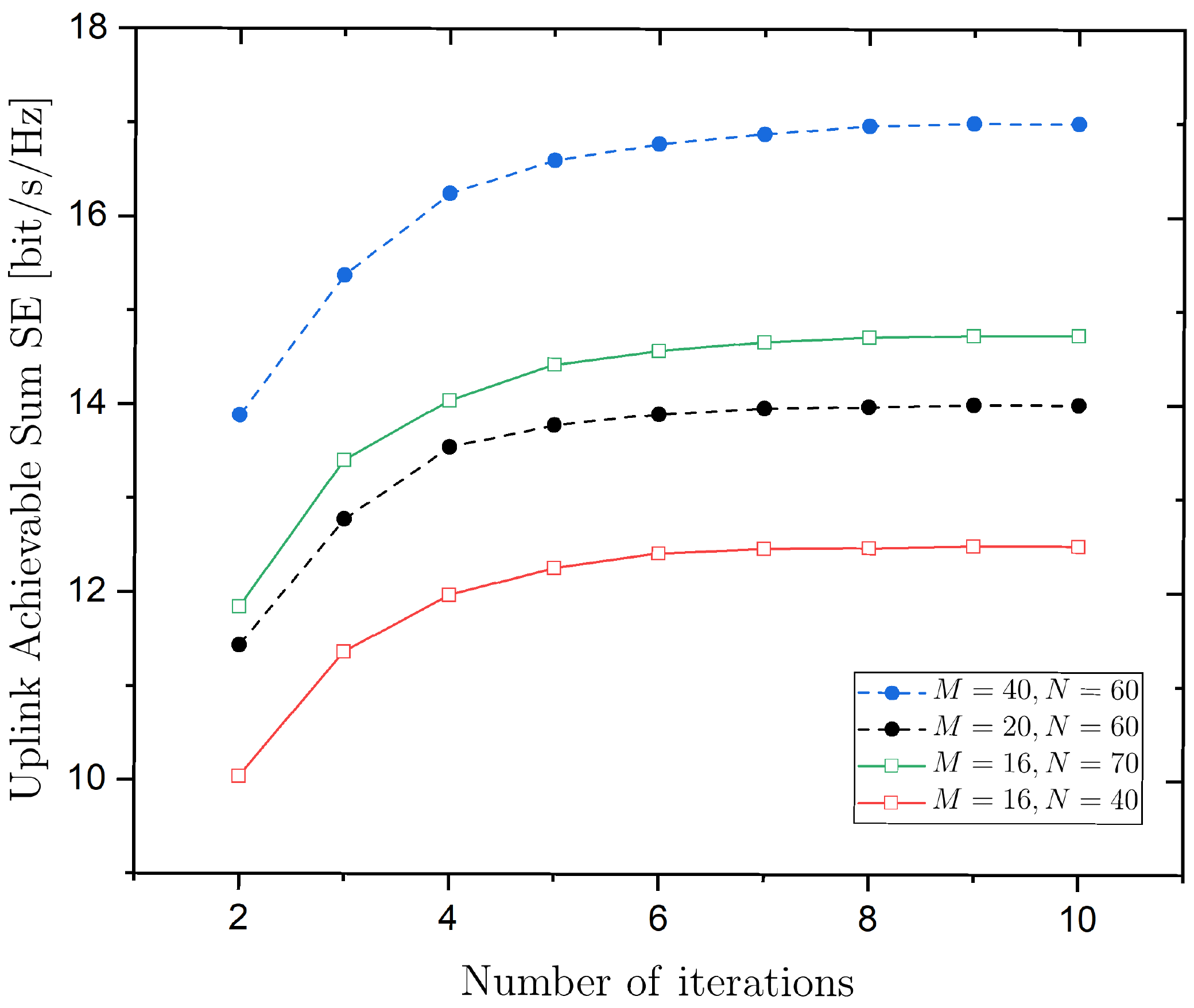}}\qquad
		\subfigure[]{	\includegraphics[width=0.95\linewidth]{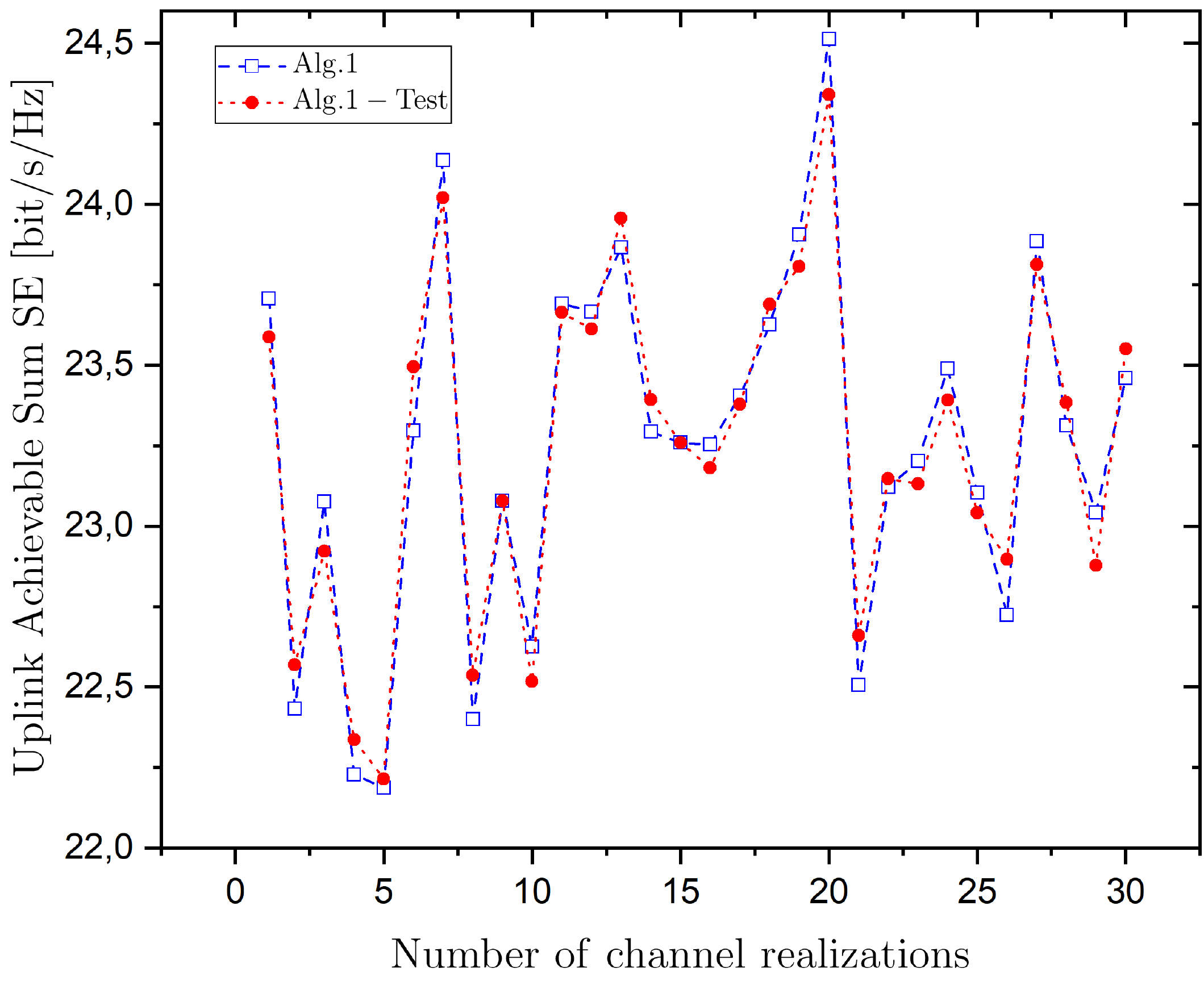}}\\
		\caption{ Uplink achievable sum SE of an IRS-assisted MIMO system with imperfect CSI  versus: (a) the number of iterations $N$ ( $ K=5 $, $ p=0\mathrm{dB }$, $ \kappa_{\mathrm{BS}}=\kappa_{\mathrm{UE}}=0.126^2 $) for varying BS antennas $ M $ and IRS elements $ N $; (b) $ 30 $ channel realizations   ($ M=16 $, $ N=20 $, $ K=5 $, $ p=20\mathrm{dB }$, $ \kappa_{\mathrm{BS}}=\kappa_{\mathrm{UE}}=0.126^{2} $).}
		\label{fig9}
	\end{figure}
	
	 In Fig. \ref{fig9}.(a), we show the convergence of the proposed algorithm, i.e., Algorithm \ref{Algoa1}. In particular, we have depicted the uplink achievable sum SE versus the number of iterations for various sets of BS antennas and IRS elements. Notably, the algorithm converges fast in all cases. For example, when $ M=20 $ and $ N=60 $, the algorithm converges in $ 7 $ iterations. Moreover, we notice that by increasing the IRS and BS sizes in terms of their elements and antennas, respectively, more iterations are required for convergence because the amount of optimization variables increases and the relevant search space is enlarged. On top of this, an increase in terms of BS antennas or IRS elements results in higher complexity of each iteration of the proposed algorithm as mentioned in Sec. \ref{IRSdesign}.
	
	 The non-convexity of the optimization problem suggests that its solution depends on the initial point, i.e., different initial points result in different locally optimal solutions. Fig. \ref{fig9}.(b) investigates this dependence on the initializations by accounting for $ 30 $ channel realizations. The initialization of Alg. 1 assumes that $ \bs^{0} =\exp\left(j\pi/2\right)\one_{N}$ as mentioned in its description. "Alg. 1-Test" in the figure assumes the best initial point out of $ 100 $ random initial points for each channel instance. The figure shows that different initializations result in different solutions and that the sum SE in both cases is almost the same, which means that this  phase shifts selection for initialization is a good choice.

	\section{Conclusion} \label{Conclusion} 
	In this paper, not only we studied the impact of both T-HWIs and IRS-HWIs on a general IRS-assisted MU-MISO system with imperfect CSI and correlated Rayleigh fading, but we also proposed a novel optimization methodology regarding the optimization of the RBM with low computational cost, being quite useful in IRS-assisted systems that have a large number of elements. In particular, we obtained the LMMSE estimate of the channel with T-HWIs and IRS-HWIs. Moreover, we derived the uplink achievable sum SE with MRC in closed form, being dependent only on large-scale statistics, and performed high computationally efficient optimization with respect to the IRS RBM. In general, we provided a methodology resulting in analytical and tractable expressions being advantageous over previous works  as shown by the simulation results. 	Furthermore, we evaluated the impact of HWIs at both the transceiver and the IRS on the system SE, and shed insightful light on their interplay with other system parameters towards efficient IRS design. Remarkably, this work opens new research directions for IRS-assisted systems such as the studies of energy efficiency and power scaling laws with imperfect CSI and T-HWIs.
	\begin{appendices}
		
		\section{Proof of Proposition~\ref{PropositionDirectChannel}}\label{Proposition1}
		According to \cite[Ch. 12]{Kay} The LMMSE estimator of $ \bh_{k} $ is obtained by $ \hat{\bh}_{k} =\bF\br_{k}$, 
		where $\bF $ is derived my minimizing $ \tr\big(\EE\big[\left(\hat{\bh}_{k}-{\bh}_{k}\right)\left(\hat{\bh}_{k}-{\bh}_{k}\right)^{\H}\big]\big) $ as
		\begin{align}
			\bF=\EE\left[\br_{k}\bh_{k}^{\H}\right]\left(\EE\left[\br_{k}\br_{k}^{\H}\right]\right)^{-1}.\label{Cor6}
		\end{align}
		Given that the additive distortions and the receiver noise are uncorrelated with overall channel $ \bh_{k} $, the first term of \eqref{Cor6} becomes
		\begin{align}
			\EE\left[\br_{k}\bh_{k}^{\H}\right]&=\EE\big[\big(\bh_{k}+ \sum_{i=1}^{K}\frac{ \tilde{\delta}_{\mathrm{t},i}}{\tau P}\bh_{i}+ \frac{\tilde{\deltav}_{\mathrm{r}} +\bw_{k}}{ \tau P}\big)\bh_{k}^{\H}\big]\\
			&=\EE\left[\bh_{k}\bh_{k}^{\H}\right]=\bR_{k}.\label{Cor0}
		\end{align}
		Regarding the second term, we have 
		\begin{align}
			\EE\left[\br_{k}\br_{k}^{\H}\right]&=\bR_{k}+ \frac{\kappa_{\mathrm{UE}}}{\tau }\sum_{i=1}^{K} \bR_{i}+ \frac{\kappa_{\mathrm{BS}}}{\tau}\!\sum_{i=1}^{K}\!\Id_{M}\!\circ\! \bR_{i} +\frac{\sigma^2}{ \tau P }\Id_{M},\label{Cor1}
		\end{align}
		where we have taken into account that the additive distortions and the receiver noise are uncorrelated with each other. Also, we have used that the variance of $ \tilde{\deltav}_{\mathrm{r}} $ is $ \tau P { \bm \Upsilon}$ with ${ \bm \Upsilon}=\kappa_{\mathrm{BS}}P\sum_{i=1}^{K} \tr\left(\bR_{i}\right)$.
		As a result, the LMMSE estimate is given by inserting \eqref{Cor0} and \eqref{Cor1} into \eqref{Cor6} as 
		\begin{align}
			\hat{\bh}_{k}=\bR_{k}\bQ_{k}\br_{k}.
		\end{align}
		Furthermore, the covariance matrix of the estimated channel is obtained as
		\begin{align}
			\EE\left[\hat{\bh}_{k}	\hat{\bh}_{k}^{\H}\right]=\bR_{k}\bQ_{k}\bR_{k}.
		\end{align} 
		\section{Proof of Theorem~\ref{theorem:ULDEMMSE}}\label{theorem1}
		For the derivation of $\gamma_{k}$ for finite $M$, we recall each term of~\eqref{sig11} and \eqref{int1}. Generally, we are going to apply a useful property suggesting that $\bx^{\H}\by = \tr(\by \bx^{\H})$ for any vectors $\bx$, $\by$. First, we obtain the $ S_{k} $ given by~\eqref{sig11}. Specifically, the desired signal part $ \mathrm{DS}_{k} $ (without $ \rho_{k} $) is written as
		\begin{align}
			\EE\big[
			\hat{\bh}_{k}^{\H}{\bh}_{k}\big]&=\tr\big( \EE\left[{\bh}_{k} \hat{\bh}_{k}^{\H} \right] \big) \\
			&=\tr\left( \EE\left[{\bh}_{k} \br_{k}^{\H}\bQ_{k} \bR_{k}\right] \right)\label{term1}\\
			&=\tr\left(\bPsi_{k}\right)\label{term2},
		\end{align}
		where, in \eqref{term1}, we have used \eqref{estim1}, while the last step is obtained by applying the expectation between $ \bh_{k} $ and $ \br_{k} $ and by considering that the mean value of $ \tilde{\delta}_{\mathrm{t},k} $ is zero.
		
		Regarding the second-order moment in the denominator, expressing the MU interference part $ \mathrm{MUI}_{ik} $ for $i\ne k$, it is written as
		\begin{align}
			&\EE\big[ \big| \hat{\bh}_{k}^{\H}{\bh}_{i}\big|^{2}\big]=\tr\!\left(\bPsi_{k} \bR_{i}\right),\label{54}
		\end{align}
		which relies on the   independence between the two random vectors.
		
		
		For the power of the beamforming uncertainty, we have 
		\begin{align}
			\EE\left\{|\mathrm{BU}_{k}|^{2}\right\} &=\EE\big[ \big| \hat{\bh}_{k}^{\H}{\bh}_{k}-\EE\big[
			\hat{\bh}_{k}^{\H}{\bh}_{k}\big]\big|^{2}\big]\label{est0}\\
			&=\EE\big[ \big| \hat{\bh}_{k}^{\H}{\bh}_{k}\big|^{2}\big]-\big|\EE\big[
			\hat{\bh}_{k}^{\H}{\bh}_{k}\big]\big|^{2} \label{est2}\\
			&=\EE\big[ \big| \hat{\bh}_{k}^{\H} \hat{\bh}_{k} +\hat{\bh}_{k}^{\H}\tilde{\bh}_{k}\big|^{2}\big]-\big|\EE\big[
			\hat{\bh}_{k}^{\H}\hat{\bh}_{k}\big]\big|^{2}\label{est3} \\
			&=\EE\big[|\hat{\bh}_{k}^{\H}\tilde{\bh}_{k}|^{2}\big] \label{est5}\\
			&=\tr\!\left(\bPsi_{k} \bR_{k}\right)-\tr\left( \bPsi_{k}^{2}\right),\label{est4}
		\end{align}
		where in~\eqref{est3}, we have substituted \eqref{current}, and in~\eqref{est3}, we have applied the property $\EE\left[ |X+Y|^{2}\right] =\EE\left[ |X|^{2}\right] +\EE\left[ |Y^{2}|\right]$, which holds between two independent random variables when one of them has zero mean value, e.g., $\EE\left[ X\right]=0 $. Equation \eqref{est4} follows from the facts that $ \EE\big[|\hat{\bh}_{k}^{\H}\tilde{\bh}_{k}|^{2}\big]=\bPsi_{k}\left(\bR_{k}-\bPsi_{k}\right) $ by taking advantage of the independence between the two random vectors.
		The derivation of the term $\sigma_{\mathrm{UE}}^{2}$, corresponding to the additive transmit distortion from all UEs, is straightforward. Specifically, we have
		\begin{align}
			\sigma_{\mathrm{UE}}^{2}&= \sum_{i=1}^{K}\rho_{i}\kappa_{\mathrm{UE}}\EE\big\{|\bv_{k}^{\H}\bh_{i}|^{2} \big\}\label{simgaue1}\\
			&=\kappa_{\mathrm{UE}}\left( \rho_{k}\EE\big\{|\bv_{k}^{\H}\bh_{k}|^{2} \right\} +\sum_{i\ne k}^{K}\rho_{i}\EE\left\{|\bv_{k}^{\H}\bh_{i}|^{2} \right\}\!\big).\label{simgaue}
		\end{align}
		
		In \eqref{simgaue1}, we have taken the expectation with respect to the transmit distortion for fixed channel realizations. The first part in \eqref{simgaue} is obtained as
		\begin{align}
			\EE\left\{|\bv_{k}^{\H}\bh_{k}|^{2} \right\}&=\EE\big[ \big| \hat{\bh}_{k}^{\H} \hat{\bh}_{k} +\hat{\bh}_{k}^{\H}\tilde{\bh}_{k}\big|^{2}\big]\\
			&=\EE\big[|\hat{\bh}_{k}^{\H}\hat{\bh}_{k}|^{2}\big]+\EE\big[|\hat{\bh}_{k}^{\H}\tilde{\bh}_{k}|^{2}\big]\\
			&=|\tr\!\left(\bPsi_{k}\right)|^{2}+\tr\!\left(\bPsi_{k} \bR_{k}\right)-\tr\left( \bPsi_{k}^{2}\right),\label{esto1}
		\end{align}
		where we have used similar steps as before. Hence, inserting \eqref{esto1} into \eqref{simgaue}, and noticing that the second part of \eqref{simgaue} is identical to $ \mathrm{MUI}_{ik} $, $ \sigma_{\mathrm{UE}}^{2} $ becomes 
		\begin{align}
			\!\!\sigma_{\mathrm{UE}}^{2}\!=\!\kappa_{\mathrm{UE}}\big(\rho_{k}|\tr\!\left(\bPsi_{k}\right)|^{2}+\sum_{i=1}^{K}\!\rho_{i}  \tr\!\left(\bPsi_{k} \bR_{i}\right)-\rho_{k}\tr\left( \bPsi_{k}^{2}\right)\!\big)\!.\label{sue1}
		\end{align}
		
		The term $\sigma_{\mathrm{BS}}^{2}$, concerning the additive receive distortion at the BS, is obtained as
		\begin{align}
			\sigma_{\mathrm{BS}}^{2}
			&=\bkappa_{\mathrm{BS}}\EE\big\{\hat{\bh}_{k}^{\H}\big(\sum_{i=1}^{K}\rho_{i} \Id_{M}\circ \bh_{i}\bh_{i}^{\H}\big)\hat{\bh}_{k}\big\}\label{bs1} \\
			&=\underbrace{\bkappa_{\mathrm{BS}}\EE\big\{\hat{\bh}_{k}^{\H}\big(\rho_{k} \Id_{M}\circ \bh_{k}\bh_{k}^{\H}\big)\hat{\bh}_{k}\big\}}_{\mathcal{I}_1}\nn\\&+\underbrace{\bkappa_{\mathrm{BS}}\EE\big\{\hat{\bh}_{k}^{\H}\big(\sum_{i\ne k}^{K}\rho_{i} \Id_{M}\circ \bh_{i}\bh_{i}^{\H}\big)\hat{\bh}_{k}\big\}}_{\mathcal{I}_2}\label{bs2},
		\end{align}
		where, 
		in \eqref{bs1}, we have taken the expectation with respect to the receive distortion for fixed channel realizations, we have accounted for MRC, i.e, $ \bv_{k}=\bh_{k} $, and we have used the Hadamard product to write the diagonal matrix. In the next equation, we have simply split the sum and denoted the two parts as $ \mathcal{I}_1 $ and $ \mathcal{I}_2 $. In the case of the former part, we have
		\begin{align}
			\mathcal{I}_1&=\bkappa_{\mathrm{BS}}\EE\!\big\{ \!\tr\! \big(\!\big(\rho_{k} \Id_{M}\circ \big(\hat{\bh}_{k}\hat{\bh}_{k}^{\H}+\tilde{\bh}_{k}\tilde{\bh}_{k}^{\H}\big)\!\big)\!\hat{\bh}_{k}\hat{\bh}_{k}^{\H}\big)\!\!\big\}\label{key}\\
			&=\underbrace{\bkappa_{\mathrm{BS}} \EE\!\big\{ \!\tr\! \big(\!\big(\rho_{k} \Id_{M}\circ \hat{\bh}_{k}\hat{\bh}_{k}^{\H}\!\big)\!\hat{\bh}_{k}\hat{\bh}_{k}^{\H}\big)\!\!\big\}}_{\mathcal{I}_{11}}\nn\\&+\underbrace{\bkappa_{\mathrm{BS}} \EE\!\big\{ \!\tr\! \big(\!\big(\rho_{k} \Id_{M}\circ \tilde{\bh}_{k}\tilde{\bh}_{k}^{\H}\!\big)\!\hat{\bh}_{k}\hat{\bh}_{k}^{\H}\big)\!\!\big\}}_{\mathcal{I}_{12}},\label{I1}
		\end{align}
		where, at first, we have used \eqref{current} and exploited that $ \hat{\bh}_{k} $ and $ \tilde{\bh}_{k} $ are uncorrelated. Next, we have split \eqref{I1} into $ \mathcal{I}_{11} $ and $ \mathcal{I}_{12} $. For $ \mathcal{I}_{11} $, we exploit that the diagonal matrix can be written as $ \Id_{M}\circ \hat{\bh}_{k}\hat{\bh}_{k}^{\H}= \sum_{m=1}^{M}|\bee_{m}^{\H} \hat{\bh}_{k} |^{2}\bee_{m}\bee_{m}^{\H}$, where $ \bee_{m} $ is the $ m $th column of $ \Id_{M} $ \cite{Bjornson2015}. Hence, we have
		\begin{align}
		&	\mathcal{I}_{11}=\bkappa_{\mathrm{BS}} \rho_{k}\sum_{m=1}^{M} \EE\!\big\{ \!\tr\! \big(\bee_{m}\bee_{m}^{\H} \hat{\bh}_{k}\hat{\bh}_{k}^{\H} \bee_{m}\bee_{m}^{\H}\hat{\bh}_{k}\hat{\bh}_{k}^{\H}\big)\!\!\big\}\label{I111}\\
			&=\bkappa_{\mathrm{BS}}\rho_{k}\sum_{m=1}^{M} \EE\!\big\{\! | \hat{\bh}_{k}^{\H} \bee_{m}\bee_{m}^{\H} \hat{\bh}_{k}|^{2}\!\big\}\label{I112}\\
			&=\bkappa_{\mathrm{BS}}\rho_{k} \sum_{m=1}^{M} \big(| \tr \big( \bee_{m} \bee_{m}^{\H } \bPsi_{k}\big)|^{2}\! +\!\tr \big(\bee_{m} \bee_{m}^{\H } \bPsi_{k} \bee_{m} \bee_{m}^{\H } \bPsi_{k}\big)\big)\label{I113}\\
			&=\bkappa_{\mathrm{BS}} \rho_{k}\big(| \tr \big( \Id_{M}\circ \bPsi_{k}\big)|^{2}\! +\!\tr \big(\big(\Id_{M}\circ \bPsi_{k}\big) \bPsi_{k}\big)\big)\label{I114},
		\end{align}
		where, in \eqref{I113}, we have used \cite[Lemma~2]{Bjornson2015}. In the next equation, we have reverted the matrix expansion. More easily, in the case of $ \mathcal{I}_{12} $, we have
		\begin{align}
			\mathcal{I}_{12}&=\bkappa_{\mathrm{BS}} \rho_{k}\sum_{m=1}^{M} \EE\!\big\{ \!\tr\! \big(\bee_{m}\bee_{m}^{\H} \tilde{\bh}_{k}\tilde{\bh}_{k}^{\H} \bee_{m}\bee_{m}^{\H}\hat{\bh}_{k}\hat{\bh}_{k}^{\H}\big)\!\!\big\}\label{I121}\\
			&=\bkappa_{\mathrm{BS}}\rho_{k} \tr \big(\!\big( \Id_{M}\circ \big(\bR_{k}-\bPsi_{k}\big)\!\big)\bPsi_{k}\big)
			\label{I122},
		\end{align}
		where, in the first equation, we have used the diagonal matrix expansion, and in the second equation, we have reverted this expansion after taking advantage of the independence between $ \tilde{\bh}_{k} $ and $ \hat{\bh}_{k} $. Substitution of \eqref{I114} and \eqref{I122} into \eqref{I1} gives $\mathcal{I}_1 $ after simple algebraic manipulations as
		\begin{align}
			\mathcal{I}_1=\bkappa_{\mathrm{BS}}\rho_{k} \big(| \tr \big( \Id_{M}\circ \bPsi_{k}\big)|^{2}\! +\!\tr \big(\big(\Id_{M}\circ \bR_{k}\big) \bPsi_{k}\big)\big)\label{i1final}.
		\end{align}
		
		Regarding $ \mathcal{I}_2 $, it follows that
		\begin{align}
			\mathcal{I}_2&=\bkappa_{\mathrm{BS}}\EE\big\{\!\!\tr\!\big(\!\!\big(\sum_{i\ne k}^{K}\rho_{i} \Id_{M}\circ \bh_{i}\bh_{i}^{\H}\!\big)\!\hat{\bh}_{k}\hat{\bh}_{k}^{\H}\big)\!\!\big\}\label{bs3}\\
			&=\bkappa_{\mathrm{BS}}\sum_{i\ne k}^{K}\rho_{i}\tr\big(\EE\big\{ \left(\Id_{M}\circ \bh_{i}\bh_{i}^{\H}\right)\!\hat{\bh}_{k}\hat{\bh}_{k}^{\H}\big\}\big)\label{bs4}\\
			&=\bkappa_{\mathrm{BS}}\sum_{i\ne k}^{K}\rho_{i}\tr\left(\!\left(\Id_{M}\circ \bR_{i} \right)\!\bPsi_{k}\right).\label{bs5}
		\end{align}
		
		In the first equality, $ \mathcal{I}_2 $ has been written in terms of the trace. Next, we have exchanged the orders among summation, trace, and expectations since they are linear operators. In \eqref{bs5}, given the independence between indices $ i $ and $ k $, we have computed the separate expectations.
		Having obtained $ \mathcal{I}_1 $ and $ \mathcal{I}_2 $, we replace them in \eqref{bs2} to obtain $ \sigma_{\mathrm{BS}}^{2} $ as
		\begin{align}
			\sigma_{\mathrm{BS}}^{2}=\bkappa_{\mathrm{BS}} \big(\rho_{k}| \tr \left( \Id_{M}\circ \bPsi_{k}\right)|^{2}\! +\sum_{i=1}^{K}\!\rho_{i}\tr \left(\left(\Id_{M}\circ \bR_{i}\right) \bPsi_{k}\right)\big)\label{i1final1}.
		\end{align}
		
		In the case of $ \EE\left\{|\mathrm{RN}_{k}|^{2}\right\} $, we easily obtain
		\begin{align}
			\EE\left\{|\mathrm{RN}_{k}|^{2}\right\}
			&=\sigma^2 \tr\!\left(\bPsi_{k}\right).\label{th2}
		\end{align}
		
		Use of \eqref{term2}, \eqref{est4}, \eqref{sue1}, \eqref{i1final1}, and \eqref{th2} concludes the proof by resulting in $ \bar{S}_{k} $ and $ \bar{I}_{k} $.
		
		\section{Proof of Proposition~\ref{Prop:optimPhase}}\label{optimPhase}
		We aim at finding the gradient of $	\mathcal{R} $ with respect to $ \phi_{n}, n=1, \ldots, N$. We use the facts that $ \pdv{	\mathcal{R}}{\phi_{n}^{*}}	=\frac{\pdv{ \gamma_{k}}{\phi_{n}^{*}}}{\ln\left(2\right)\left(1+\frac{S_{k}}{I_{k}}\right)} $, which 
		requires the derivation of $ \pdv{\gamma_{k}}{\phi^{*}_{n}} $. A closer observation of Theorem \ref{theorem:ULDEMMSE}, providing $ \gamma_{k} $, reveals that it is a fraction consisting of terms including functions of traces. Hence, the standard quotient rule derivative gives
		\begin{align}
			\pdv{\gamma_{k}}{\phi^{*}_{n}}=\frac{\pdv{S_{k}}{\phi^{*}_{n}}I_{k}-S_{k}\pdv{I_{k}}{\phi^{*}_{n}}}{I_{k}^{2}},\label{gam1}
		\end{align}
		where the partial derivatives follow. Henceforth, for the sake of simplicity, we replace the notation for the partial derivative with respect to $ \phi^{*}_{n} $ by $ \left(\cdot\right)' $. Specifically, in the case of $ S_{k}' $, we obtain
		\begin{align}
			S_{k}'&=\rho_{k}\left(|\tr\left(\bPsi_{k}\right)|^{2}\right)'\\
			&=2\rho_{k}\tr\left(\bPsi_{k}\right)\tr\left(\bPsi_{k}'\right),\label{sk1}
		\end{align}
		where \eqref{sk1} includes a simple derivative. Since all the terms in $	\bPsi_{k}$ depend on $ \phi^{*}_{n} $, we have
		\begin{align}
&\!\!\!\!			\tr\left(\bPsi_{k}'\right)=\tr\left(\bR_{k}'\bQ_{k}\bR_{k}+\bR_{k}\bQ_{k}'\bR_{k}+\bR_{k}\bQ_{k}\bR_{k}'\right)\label{sk2}\\
			&\!\!\!\!=\tr\!\big(\bR_{k}'\bQ_{k}\bR_{k}-\bR_{k}\bQ_{k}\left(\bQ_{k}^{-1}\right)^{'}\bQ_{k}\bR_{k}+\bR_{k}\bQ_{k}\bR_{k}'\big),\label{sk3}
		\end{align}
		where the derivative of $\bQ_{k}   $, being an inverse matrix, is obtained by \cite[Eq. 40]{Petersen2012} while $ \left(\bQ_{k}^{-1}\right)^{'}  $ is obtained as
		\begin{align}
		\!\!	\left(\bQ_{k}^{-1}\right)^{'}
			&\!=\!\bR_{k}'\!+\! \frac{\kappa_{\mathrm{UE}}}{\tau }\!\sum_{i=1}^{K} \bR_{i}'\!+\!\frac{ \kappa_{\mathrm{BS}}\!\sum_{i=1}^{K}\!\left(\Id_{M}\!\circ \!\bR_{i} \right)'}{ \tau }.\label{sk4}
		\end{align}
		
		Substitution of \eqref{sk4} into \eqref{sk3} gives
		\begin{align}
&		\tr\left(\bPsi_{k}'\right)=\tr\left(\bR_{k}'\bQ_{k}\bR_{k}\right)-\tr\big(\big(\bR_{k}\bQ_{k}\big)^{2}\bR_{k}'\big)\nn\\
&\!-\!\frac{\kappa_{\mathrm{UE}}}{\tau }\!\sum_{i=1}^{K}\tr\!\big(\!\big(\bR_{k}\bQ_{k}\big)^{2}\bR_{i}'\big)\!-\!\frac{ \kappa_{\mathrm{BS}}}{ \tau }\!\sum_{i=1}^{K}\!\tr\!\big(\!\big(\bR_{k}\bQ_{k}\big)^{\!2}\big(\Id_{M}\!\circ\! \bR_{i} \big)'\big)\nn\\
&+\tr\left(\bR_{k}\bQ_{k}\bR_{k}'\right).\label{sk5}
		\end{align}
		
		To proceed further, for the sake of exposition, let the matrices $ \bA, \bB, \bC \in\mathbb{C}^{M \times M} $,
		we have denoted
		\begin{align}
			&\bL\! \left(\bA,\bB, \bC\right)=\al\beta_{2,i}\!\left[\bH_{1}^{\H}\bA\bQ_{k} \bH_{1} \bTheta\bR_{\mathrm{IRS},i}\right]_{n,n}\nn\\
			&
			-\big(1\!+\!\frac{\kappa_{\mathrm{UE}}}{\tau }\big)\al\beta_{2,k}\!\left[\bH_{1}^{\H}\bA \bQ_{k}\bC \bQ_{k} \bA\bH_{1} \bTheta\bR_{\mathrm{IRS},k}\right]_{n,n}\nn\\
			&\!\!\!-\!\al\frac{ \kappa_{\mathrm{BS}}}{ \tau }\!\sum_{i=1}^{K}\beta_{2,i}\!\left[\bH_{1}^{\H}\bA \bQ_{k}\bC \bQ_{k} \bA\bH_{1} \bTheta\bR_{\mathrm{IRS},i}\right]_{n,n}\!
			\nn\\
			&+\!\al\beta_{2,k}\!\left[\bH_{1}^{\H}\bB \bQ_{k} \bH_{1} \bTheta\bR_{\mathrm{IRS},k}\right]_{n,n}.\label{LABC}
		\end{align}
		
		Also, we are going to use the following useful lemma.
		\begin{lemma}\label{traceProd}
			Let 	 $ \bA \in\mathbb{C}^{M\times M} $ be independent of $ \bTheta$ and $\bR_{k}= \beta_{2,k}\bH_{1} \bTheta$ $\bR_{\mathrm{IRS},k}\bTheta^{\H}\bH_{1}^{\H} $, then
			\begin{align}
				\tr\left( \!\!\bA\pdv{\bR_{k}}{\phi^{*}_{n}}\!\right) =\al\beta_{2,k}[\bH_{1}^{\H}\bA\bH_{1} \bTheta\bR_{\mathrm{IRS},k}]_{n,n}.
		\end{align}\end{lemma}
		\proof We have
		\begin{align}
			\!\tr\!\left(\! \!\bA\pdv{\bR_{k}}{\phi^{*}_{n}}\!\right) &\!=\!\sum_{i,j}[\bA]_{ij}\pdv{[\bR_{k}]_{ji}}{\phi^{*}_{n}}\\
			&\!=\!\al\beta_{2,k}\sum_{i,j}[\bA]_{ij}[\bH_{1} \bTheta\bR_{\mathrm{IRS},k}]_{jn}[\bH_{1}^{\H} ]_{in}^{\T}\\
			&\!=\!\al\beta_{2,k}[\bH_{1}^{\H}\bA\bH_{1} \bTheta\bR_{\mathrm{IRS},k}]_{nn},
		\end{align}
		since $ \pdv{[\bR_{k}]_{ji}}{\phi^{*}_{n}}=\al\beta_{2,k} [\bH_{1} \bTheta\bR_{\mathrm{IRS},k}]_{jn}[\bH_{1}^{\H} ]_{in}^{\T} $.
		\endproof
		By exploiting Lemma \ref{traceProd} for each term of \eqref{sk5} and that the trace of the transpose of a matrix equals the trace of this matrix, after several algebraic manipulations, we obtain 
		\begin{align}
			\!\!\tr\left(\bPsi_{k}'\right)&\!=\!\al\beta_{2,k}\!\big[\bH_{1}^{\H}\bR_{k}\bQ_{k}\!\big(\!2\Id_{M}\!-\!\big(1\!+\!\frac{\kappa_{\mathrm{UE}}}{\tau }\big)\!\bR_{k}\bQ_{k}\!\big)\!\bH_{1} \bTheta\bR_{\mathrm{IRS},k}\big]_{n,n}\nn\\
			&-\!\al\frac{ \kappa_{\mathrm{BS}}}{ \tau }\!\sum_{i=1}^{K}\beta_{2,i}\!\big[\bH_{1}^{\H}\!\left(\bR_{k}\bQ_{k}\!\right)^{2}\!\bH_{1} \bTheta\bR_{\mathrm{IRS},i}\big]_{n,n}\!\\
			&=\bL\! \left(\bR_{k},\bR_{k},\Id_{M} \right)
			\label{sk6},
		\end{align}
		which completes the derivation of $ S_{k}' $ after its insertion into \eqref{sk1}. 
		
		The computation of $ I_{k}' $ consists of the sum of the derivatives of different terms, requiring separate manipulations. Thus, we start by the computation of the derivative of the first term in \eqref{Den1} as
		\begin{align}
&\!\!			\left(\tr\!\left(\bPsi_{k} \bR_{i}\right)\right)'=\tr\!\left(\bPsi_{k}' \bR_{i}+\bPsi_{k} \bR_{i}'\right)\\
			&\!\!=\bL\! \left(\bR_{i}\bR_{k},\bR_{i}\bR_{k},\bR_{i} \right)+\al\beta_{2,k}[\bH_{1}^{\H}\bPsi_{k}\bH_{1} \bTheta\bR_{\mathrm{IRS},k}]_{n,n},\label{I6}
		\end{align}
		where we have applied Lemma \ref{traceProd}. \
		
		Moreover, we have 
		\begin{align}
			\left(\tr\left( \bPsi_{k}^{2}\right)\right)'&=2\tr\!\left(\bPsi_{k}' \bPsi_{k}\right)\\
			&=2\bL\! \left(\bPsi_{k}\bR_{k},\bPsi_{k}\bR_{k},\bPsi_{k}\right).\label{I8}
		\end{align}
		
		In addition, we have
		\begin{align}
			\left(| \tr \left( \Id_{M}\circ \bPsi_{k}\right)|^{2}\right)'&=2 \tr \left( \Id_{M}\circ \bPsi_{k}\right) \tr \left( \Id_{M}\circ \bPsi_{k}\right)'\\
			&=2 \tr \left( \Id_{M}\circ \bPsi_{k}\right) \bL\! \left(\bR_{k},\bR_{k}, \Id_{M}\right).\label{I9}
		\end{align}
		
		Furthermore, we obtain
		\begin{align}
			&\!\!\left(\tr\! \left(\!\left(\Id_{M}\!\circ\! \bR_{i}\right) \!\bPsi_{k}\right)\!\right)'\!=\!\tr\!\left(\! \left(\!\left(\Id_{M}\!\circ\! \bR_{i}\right)'\! \bPsi_{k}\!+ \!\left(\Id_{M}\!\circ \!\bR_{i}\right)\!\bPsi_{k}'\right)\!\right)\\
			&\!\!=\al\beta_{2,k}[\bH_{1}^{\H}\bPsi_{k}\bH_{1} \bTheta\bR_{\mathrm{IRS},k}]_{n,n}\nn\\
			&\!\!+\bL\! \left(\left(\Id_{M}\circ \bR_{i}\right)\bR_{k},\left(\Id_{M}\circ \bR_{i}\right)\bR_{k},\Id_{M}\circ \bR_{i}\right),\label{I10}
		\end{align}
		where we have used Lemma \ref{traceProd}. Equations \eqref{I6}-\eqref{I10} give $ I_{k}' $, which together with $ S_{k}' $ provide $ \pdv{\gamma_{k}}{\phi^{*}_{n}} $.
	\end{appendices}

	\bibliographystyle{IEEEtran}
	
	\bibliography{mybib,10}

\end{document}